\documentclass[12pt]{article}

\usepackage[font=small,format=plain,labelfont=bf,textfont=it]{caption}
\usepackage{bm}
\usepackage{setspace}
\usepackage{apptools}

\usepackage{fullpage, verbatim, amsfonts, amsmath, amssymb, amsthm, latexsym, setspace, graphicx}
\allowdisplaybreaks

\usepackage[T1]{fontenc}
\usepackage[latin9]{inputenc}
\usepackage{geometry}
\usepackage{xcolor}
\usepackage{natbib}
\usepackage{overpic}

\usepackage[french, english]{babel}
\usepackage{hyperref,bookmark}

\hypersetup{linkcolor=MyDarkBlue, citecolor=MyDarkBlue, colorlinks=true} 
	\definecolor{MyDarkBlue}{rgb}{0,0.08,0.45}
	\definecolor{MyDarkGreen}{rgb}{0,0.55,0.08}

\theoremstyle{plain}
\newtheorem{thm}{Theorem}
\newtheorem{lem}{Lemma}
\newtheorem{cor}{Corollary}
\newtheorem{prop}{Proposition}
\newtheorem{defn}{Definition}

\newtheorem{assume}{Assumption}

\newtheorem{ex}{Example}

\theoremstyle{definition}
\newtheorem{rem}{Remark}

\AtAppendix{\counterwithin{lem}{section}}
\AtAppendix{\counterwithin{prop}{section}}
\AtAppendix{\counterwithin{thm}{section}}
\AtAppendix{\counterwithin{cor}{section}}
\AtAppendix{\counterwithin{ex}{section}}
\AtAppendix{\counterwithin{assume}{section}}

\newcommand{\argmax}{\mathop{\mathrm{argmax}}}
\newcommand{\argmin}{\mathop{\mathrm{argmin}}}

\newcommand{\cP}{{\cal P}}
\newcommand{\cC}{{\cal C}}

\newcommand\cites[1]{\citeauthor{#1}'s \citeyearpar{#1}}

\newcommand{\finex}{\leavevmode\unskip\penalty9999 \hbox{}\nobreak\hfill\quad\hbox{$\blacktriangle$}}

\begin{document}

	\title{Contagious Ambiguity\thanks{Frick: Princeton University (mfrick@princeton.edu); Iijima: Princeton University (riijima@princeton.edu); Oyama: University of Tokyo (oyama@e.u-tokyo.ac.jp). For helpful comments, we thank Roberto Corrao, Drew Fudenberg, Faruk Gul, Atsushi Kajii, Peter Klibanoff, Elliot Lipnowski, Stephen Morris, Sujoy Mukerji, Pietro Ortoleva, Ludovic Renou, Tomasz Strzalecki, Takashi Ui, Mu Zhang, as well as numerous seminar and conference audiences. We also gratefully acknowledge the financial support from Sloan Research Fellowships (Frick and Iijima) and from JSPS KAKENHI Grant 23K01303 (Oyama).} }
	\author{Mira Frick \and Ryota Iijima \and Daisuke Oyama}
	\date{August 20, 2026}

	\maketitle

	\begin{abstract}
	
	We study the strategic impact of ambiguity through the channel of higher-order beliefs. We show that even small amounts of prior ambiguity about game payoffs can generate arbitrarily large amounts of higher-order ambiguity. This gives rise to a novel form of contagion: vanishingly small payoff ambiguity can select ``secure'' actions (e.g., non-participation) as the unique equilibrium outcome even when those actions are almost dominated. We highlight two main implications. First, classical robustness results under probabilistic uncertainty break down under ambiguity when players can deviate to secure actions. Second, if a designer can introduce small amounts of payoff ambiguity into a game, this can serve as a powerful tool for unique implementation.

		\end{abstract}

		\newpage

	\onehalfspacing

\section{Introduction}

In many economic environments, agents face ambiguity, i.e., are unable to form a precise probabilistic belief about fundamentals. For example, investors may have ambiguous beliefs about the profitability of an investment, or consumers may face ambiguity about the quality of a new technology. Many such settings naturally feature strategic externalities (e.g., through synergies or network effects). In this case, agents' behavior may depend not only on their first-order ambiguity (i.e., ambiguous beliefs about fundamentals) but also on their higher-order ambiguity (i.e., ambiguous beliefs about other agents' beliefs). 

In this paper, we study the strategic impact of ambiguity through the channel of higher-order beliefs. We show that even small amounts of prior ambiguity about game payoffs can generate arbitrarily large amounts of higher-order ambiguity. This gives rise to a novel form of contagion: vanishingly small payoff ambiguity can select ``secure'' actions (e.g., non-participation) as the unique equilibrium outcome even when those actions are almost dominated. We highlight two main implications. First, classical robustness results under probabilistic uncertainty break down under ambiguity when players can deviate to secure actions. Second, if a designer can introduce small amounts of payoff ambiguity into a game, this can serve as a powerful tool for unique implementation. The following example illustrates the main ideas.

\subsection{Illustrative Example}\label{sec:example} 

Consider a complete-information game with two players $i \in \{1, 2\}$ and binary action sets $A_i = \{0,1\}$ for each $i$. Each player $i$'s payoff function is given by 
\[
u_i(a)= \mathbf{1}_{a_i = a_{-i} =1} -ca_i
\] 
for some $c\in (0,1)$. Thus, a player who plays $a_i =1$ always incurs a unilateral cost of $c$, and receives a benefit of $1$ if and only if the other player also plays $a_{-i} =1$. Action $a_i = 0$ is \emph{safe}, yielding a payoff of $0$ regardless of the other player's action. There are two strict Nash equilibria: $a={\bf 1}$ where both players choose action 1, and $a={\bf 0}$ where both players choose action 0. 
This stylized game captures coordination problems that are relevant in various economic settings. For example, action $a_i=1$ might represent participation in a joint investment or coordinated attack, while $a_i=0$ corresponds to non-participation in these activities. 

Observe that when $c = 0$, action $a_i = 1$ is weakly dominant for each player $i$. Thus, one might expect that when $c$ is sufficiently small (i.e., $a_i = 1$ is ``almost dominant''), equilibrium ${\bf 1}$ survives against small amounts of uncertainty about game payoffs. Under probabilistic uncertainty, this point was formalized by \cite{kajii1997}. They considered $\varepsilon$-elaborations, i.e., incomplete-information games with a common prior that assigns probability at least $1- \varepsilon$ to both players being {\it normal} types, who know that their payoffs coincide with the original game payoffs $u_i$. Their analysis implies that as long as $c < 0.5$, ${\bf 1}$ is {\it robust}, i.e., every $\varepsilon$-elaboration with small enough $\varepsilon$ admits an equilibrium where ${\bf 1}$ is played with high probability.

In this paper, we instead study the implications of small amounts of ambiguity about game payoffs. In the baseline model, we represent ambiguous beliefs using the multiple-prior approach and focus on maxmin expected utility preferences
\citep{GS1989}. Thus, we allow for {\it ambiguous $\varepsilon$-elaborations}, i.e., incomplete-information games with a common \emph{set} of priors, each of which assigns probability at least $1- \varepsilon$ to both players being normal types. We find that ${\bf 1}$ is not robust against such elaborations, because ambiguity makes ${\bf 0}$ {\it contagious}:  For every $\varepsilon>0$, there is an ambiguous $\varepsilon$-elaboration where ${\bf 0}$ is played with probability 1 in the unique equilibrium. This is the case even if $c$ is arbitrarily close to $0$.

To illustrate, consider an ambiguous elaboration of the following form. Each player $i$ has a set of types $T_i=\mathbb{N}=\{0,1,\ldots\}$. All types $t_i\geq 1$ are normal with payoffs given by $u_i$, but type $t_i=0$ has a payoff function under which action $0$ is strictly dominant. Players share a common set of priors ${\cal P}=\left\{P^\lambda:\lambda\in[0,1] \right \} \subseteq \Delta(T_1\times T_2)$. Here, for each $\lambda\in[0,1]$, the prior $P^\lambda\in\Delta(T_1\times T_2)$ is given by
\[
P^\lambda(t_1,t_2)=
\begin{cases}
\varepsilon(1-\varepsilon)^n\lambda & \text{if } (t_1,t_2)=(n+1,n) \text{ for some } n\in\mathbb{N},\\
\varepsilon(1-\varepsilon)^n(1-\lambda) & \text{if } (t_1,t_2)=(n,n+1) \text{ for some } n\in\mathbb{N},
\end{cases}
\]
and $P^\lambda(t_1, t_2) = 0$ otherwise. Thus, all priors $P^{\lambda}$ assign probability $1- \varepsilon$ to both players having payoffs $u_i$ (i.e., to the event $(t_1, t_2) \geq (1, 1)$), so there is only a small amount of prior ambiguity about game payoffs. However, there is significant ambiguity about the relative ranking of players' types, as priors differ in the probabilities $P^\lambda(t_1 \ge t_2) = \lambda$ they assign to player 1's type being higher than player 2's. This type space can be viewed as an ambiguous analog of an email-game type space \`a la \cite{rubinstein1989}, where players share the single common prior $P^{\frac{1}{2}}$.

Each player $i$ knows her own type $t_i$ and uses this information to form posteriors about her opponent's type. There are many belief-updating rules under ambiguity. For concreteness, assume for now that players perform full Bayesian updating \citep{pires2002}: $i$'s set of posteriors at type $t_i$ is given by the Bayesian updates $\{P^\lambda(\cdot  \mid t_i): \lambda\in [0,1]\}$ of all the priors conditional on $t_i$. In particular, every normal type $t_i=n \geq 1$ believes that the opponent's type $t_{-i}$ is either $n-1$ or $n+1$, and perceives the relative probability to be fully ambiguous: For each $n \geq 1$,
\[
\left\{P^\lambda \left(t_{-i} =n-1 \mid t_i = n \right) = 1 - P^\lambda \left(t_{-i} =n+1 \mid t_i = n \right) :\lambda\in[0,1] \right \}=[0,1].
\]

Note that all normal types know they are normal. Moreover, all types $t_i \geq 2$ also know that their opponent is normal, and thus have no ambiguity about either their own or their opponent's payoffs. Nevertheless, the fact that type $t_i = 1$ faces substantial ambiguity about her opponent's payoffs (which in turn induces substantial higher-order ambiguity for types $t_i \geq 2$) is enough to trigger the following contagion: First, each type $t_i =1$ chooses action $0$. This is because action $0$ yields a safe payoff of $0$, whereas the payoff to playing action $1$ is $-c$ under $i$'s worst-case posterior that assigns probability 1 to $t_{-i}= 0$ for whom action $0$ is strictly dominant.
As a result, the payoff to action $1$ is also $-c$ under the worst-case posterior of each type $t_{i} = 2$ (which assigns probability 1 to $t_{-i} = 1$), and hence $t_{i} = 2$ plays action $0$. Inductively, all types must play action $0$ in the unique equilibrium.\footnote{Indeed, more strongly, $0$ is uniquely rationalizable for all types.}

The above contagion suggests that ambiguity has dual implications for two important questions in economic theory. First, adopting a robust-prediction perspective \`a la  \cite{kajii1997} and a large follow-up literature (see Section~\ref{sec:robustness}), one can ask which equilibria of the complete-information game are robust to small amounts of \emph{exogenous} payoff uncertainty. From this perspective, the implication is that considering payoff ambiguity can greatly limit robust predictions relative to the case of probabilistic uncertainty: Even equilibria in almost dominant actions (e.g., ${\bf 1}$ when $c \approx 0$) can fail to be robust, as a small amount of payoff ambiguity can, through its effect on higher-order beliefs, lead players to opt for safe actions. For example, even a very profitable joint venture may fail to attract investors, as vanishingly small prior ambiguity about its profitability may trigger a contagion towards non-investment.

Second, a more recent literature adopts a design perspective and asks how a principal can \emph{endogenously} inject small amounts of payoff uncertainty into a game with the goal of implementing some target action profile \citep[e.g.,][]{bergemann2019, halac2025}. From this perspective, the above suggests that ambiguity can be a useful tool. Specifically, suppose the principal seeks to implement the safe action profile ${\bf 0}$ (e.g., prevent a coordinated attack) using bilateral contracts that pay each player $i$ a subsidy $w_i \geq 0$ if $i$ chooses $a_i = 0$. If the principal's goal is partial implementation, i.e., to implement ${\bf 0}$ as \emph{an} equilibrium, no subsidies are needed. However, it is natural for the principal to be concerned about strategic uncertainty, i.e., that players may coordinate on an undesirable equilibrium, and hence seek to implement ${\bf 0}$ as the \emph{unique} equilibrium.  If the principal can only offer deterministic payment schemes or introduce probabilistic uncertainty about players' subsidies, uniquely implementing ${\bf 0}$ is in general more costly than partial implementation.\footnote{Specifically, the unique implementation cost of ${\bf 0}$ under deterministic (respectively, stochastic but unambiguous) payment schemes is $1 - c$ (respectively, $\max\{1 - 2c, 0\}$).} However, suppose the principal can design {\it ambiguous payment schemes}, modeled as a set of distributions over subsidies $(w_i)$. This generates higher-order ambiguity, where each type of player $i$ knows her own realized subsidy but faces ambiguity about the other player's subsidy, their beliefs about $i$'s subsidy, etc. We will see that ambiguous payment schemes reduce the unique-implementation cost of ${\bf 0}$ to (arbitrarily close to) zero, i.e., allow the principal to eliminate strategic uncertainty for free. The idea is to introduce vanishingly small prior ambiguity about the possibility that some players receive subsidies that make $0$ dominant, and to exploit this to create a contagion towards ${\bf 0}$ as above.

\subsection{Overview of Results}\label{sec:summary}

Moving beyond this example, in the following we analyze general games with finitely many players $i \in I$ and finite action sets $A_i$ for each $i$. Section~\ref{sec:model} models players' higher-order ambiguity via an ambiguous type space, which specifies a countable set of types $T_i$ for each player, along with a compact set $\cP$ of common prior beliefs over type profiles. Each player $i$ of type $t_i$ forms a set of posterior beliefs $\cP_{t_i} $ over opponent types. We allow for fairly general belief-updating rules, including full Bayesian updating (as in the example), maximum-likelihood updating \citep{gilboa1993}, and many other updating rules proposed in the literature. 

Section~\ref{sec:main} presents the main results. Our first main result (Theorem~\ref{thm:CP-amb}) highlights a key difference between higher-order beliefs under ambiguity vs.\ probabilistic uncertainty. In unambiguous type spaces with a common singleton prior ${\cal P} = \{P\}$, \cites{kajii1997} critical path theorem implies that if players have high prior confidence in some event $E$, this limits how much higher-order uncertainty about $E$ there can be: That is, if $P(E) \approx 1$, then for all not too large $p$, the probability that players have common $p$-belief \citep{monderer1989} in $E$ must also be close to one. In stark contrast, Theorem~\ref{thm:CP-amb} shows that under ambiguity, prior confidence imposes no limits on higher-order uncertainty: For arbitrarily small $\varepsilon$, we exhibit type spaces where even though \emph{all} prior beliefs $P \in \cP$ assign probability at least $1- \varepsilon$ to an event $E$, common $\varepsilon$-belief of $E$ fails at all types. Here, we use a natural extension of common $p$-belief to ambiguous type spaces \citep{ganguli2007}. 

Second, we build on Theorem~\ref{thm:CP-amb} to study contagion in games with small amounts of payoff ambiguity and maxmin expected utility preferences (Section~\ref{sec:alphaMEU} considers more general ambiguity preferences). Generalizing the illustrative example, we find that ambiguity creates a new force for contagion: Suppose an action profile $a^*$ is \emph{secure}, i.e., each player $i$'s worst-case payoff to playing action $a^*_i$ when opponents play an arbitrary $a_{-i}$ is better than the best payoff she can achieve by deviating to any $a_i \neq a^*_i$ when opponents stick with $a^*_{-i}$. Then Theorem~\ref{prop:contagious} shows that $a^*$ is contagious under ambiguous elaborations.  Any strict Nash equilibrium in safe actions (e.g., ${\bf 0}$ in the example) is secure, but as we discuss, secure equilibria also arise in natural coordination games that do not feature safe actions. The contagion of secure equilibria is driven by the higher-order-belief effects of ambiguity in Theorem~\ref{thm:CP-amb}: Even when all priors put probability $1- \varepsilon$ on the original game payoffs and only probability $\varepsilon$ on ``crazy'' types for whom $a^*_i$ is strictly dominant, players may face substantial higher-order ambiguity about each other's beliefs regarding the probability of crazy types. When $a^*$ is secure, then as in the example, such higher-order ambiguity iteratively forces all types of all players to play $a^*$.

Section~\ref{sec:implication} illustrates the implications of these results for the dual questions highlighted in the example. First, Section~\ref{sec:robustness} considers the robustness of equilibrium predictions to small exogenous payoff uncertainty. Under probabilistic uncertainty, the literature has identified various sufficient conditions for robustness that formalize the central idea that equilibria in ``almost dominant'' actions are robust. For instance, \cite{kajii1997} show that an equilibrium $\hat a$ is robust if it is $p$-dominant for some $p < 1/|I|$ (i.e., each $\hat a_i$ is a strict best-response to any opponent strategy that puts probability at least $p$ on $\hat a_{-i}$). 
In contrast, our analysis implies that these conditions do not guarantee robustness to small payoff ambiguity. For instance, as the example illustrates, even if $\hat a$ is $p$-dominant for arbitrarily small $p > 0$, it can be eliminated by small amounts of payoff ambiguity, because by Theorem~\ref{prop:contagious} such payoff ambiguity may select a secure $a^* \neq \hat a$ as the unique equilibrium. We also discuss how to combine $p$-dominance and security to obtain stronger sufficient conditions that ensure robustness to small payoff ambiguity.

Second, Section~\ref{sec:contracts} considers how a principal who can endogenously introduce small payoff ambiguity into a game can exploit this for unique implementation. Building on Theorem~\ref{prop:contagious}, we highlight two cases in which such ambiguity can reduce the unique implementation cost of a target action profile $a^*$ all the way to its partial implementation cost. First, generalizing the example, Corollary~\ref{lem:FB bound} shows this is the case if $a^*$ is sufficiently close to secure. Second, Corollary~\ref{cor:insurance} shows that, when combined with ambiguity, a simple departure from bilateral contracts---{\it insurance contracts}---allows the principal to uniquely implement \emph{any} action profile $a^*$ at no additional cost relative to partial implementation. In contrast, in both these cases, if the principal can only introduce probabilistic payoff uncertainty, the unique implementation cost can exceed the partial implementation cost by a significant margin.

Finally, while all aforementioned results apply under general belief-updating rules, Section~\ref{sec:ML} shows that some updating rules generate even more extreme departures from the case of probabilistic uncertainty. Indeed, we provide a sense in which updating rules that are less conservative (i.e., induce smaller sets of posteriors) can magnify the higher-order belief effects of ambiguity. As an extreme example, consider maximum-likelihood updating, which refines full Bayesian updating by restricting each type $t_i$'s set of posteriors ${\cP}_{t_i}$ to the Bayesian updates of only those priors under which $t_i$ is most likely. In this case, we derive a starker failure of the critical path theorem: even a less demanding notion of common belief that we term \emph{common $\varepsilon$-justifiability} fails for arbitrarily small $\varepsilon$. Consequently, under maximum-likelihood updating, \emph{any} strict Nash equilibrium is contagious under ambiguous elaborations. Intermediate results hold for updating rules in between full Bayesian and maximum-likelihood updating.

\subsection{Related Literature}\label{sec:literature}

We contribute to the literature on ambiguous beliefs in strategic environments. Many papers study the impact of ambiguous beliefs in specific classes of games, without an explicit focus on higher-order beliefs.\footnote{See, e.g., \cite{salo1995, ellis2016, chen2026} in the context of auctions, voting, and social learning.} \cite{epstein1996} and \cite{ahn2008} provide foundations for ambiguous type spaces based on hierarchies of ambiguous beliefs.
\cite{kajii2005} formulate the equilibrium concept we use in our analysis: each player's strategy is required to be worst-case optimal against opponents' strategies given interim beliefs updated from a common set of priors.\footnote{See, e.g., \cite{billot2000, kajii2009, cerreia2022} for characterizations of the common-prior assumption under ambiguity.} Some recent papers focus on different aspects of higher-order ambiguity. \cite{huo2024} and \cite{cerreia2022} study how ambiguous beliefs affect the unique equilibrium in beauty-contest games \`a la \cite{morris2002}.  \cite{ui2025} studies binary-action global games with signals of ambiguous precision and full Bayesian updating. He shows that, unlike low-quality unambiguous signals, more ambiguous signal technologies increase the range of signals where a safe action is played, and that sufficiently ambiguous signal technologies can induce an action that is both safe and ex-ante risk-dominant to be played at all signals in the unique equilibrium.

We instead focus on contrasting the higher-order belief effects of small payoff ambiguity with those under probabilistic payoff uncertainty. Under probabilistic uncertainty, \cite{rubinstein1989} provides the seminal observation that high prior confidence in an event $E$ need not guarantee approximate common certainty of $E$, and as a result, small payoff uncertainty can eliminate some equilibria of a complete-information game by uniquely selecting another equilibrium.  Qualifying this contagion result, \cite{kajii1997} bound the amount of higher-order uncertainty that can be generated by small prior uncertainty, and use this to show that small payoff uncertainty cannot eliminate equilibria that are $p$-dominant for small enough $p$. In contrast, we show that small prior payoff ambiguity can generate arbitrarily large higher-order ambiguity, and consequently, gives rise to a more extreme form of contagion than probabilistic uncertainty: even equilibria that are arbitrarily close to dominant can be eliminated in favor of a secure equilibrium.\footnote{It is worth contrasting this finding with \cite{weinstein2007}, who instead adopt an interim perspective on robustness. Applied to the complete-information game in the example (represented by a type $t^{{\rm CK}}$ in the universal type space), their analysis implies that, even if $c \approx 0$ so that ${\bf 1}$ is almost dominant, one can slightly perturb (in the product topology) $t^{{\rm CK}}$ to some unambiguous type $t^*$ for which ${\bf 0}$ is uniquely rationalizable. However, \cites{kajii1997} analysis implies that in any type space whose common prior assigns high enough probability to the original game payoffs, such types $t^*$ must have very small prior probability. In contrast, we show that one can find ambiguous common-prior elaborations where ${\bf 0}$ is uniquely rationalizable for \emph{all} types. \cite{yokota2024} extends \cites{weinstein2007} results to perturbations of ambiguous types. }

In highlighting the role of secure actions (e.g., non-participation) as a source of contagion in games, we relate to an extensive literature \citep[going back to][]{dow1992} that proposes ambiguity as a source of market non-participation \citep[see, e.g.,][for a survey]{billot2020}. While this literature mostly focuses on non-negligible payoff ambiguity and abstracts away from strategic externalities, our analysis suggests that if one accounts for strategic externalities, even small payoff ambiguity can have a large impact through the channel of higher-order beliefs.

Our analysis complements \cite{oyama2010, oyama2012}, who consider elaborations with probabilistic but non-common priors and show that any strict Nash equilibrium is contagious with respect to such elaborations.\footnote{See also \cite{kajii1997} (pp.\ 1296-7), who note that the critical path theorem breaks down with non-common probabilistic priors.}  We highlight that ambiguity creates an alternative force for contagion that applies even when the common-prior assumption is maintained. However, reflecting the difference in mechanism, under ambiguity, the extent of contagion depends on the conservativeness of players' updating rules and the security of actions, neither of which is relevant under (common or non-common) probabilistic priors.\footnote{As Section~\ref{sec:ML} discusses, this difference relates to the novel feature under ambiguity that one must distinguish higher-order uncertainty in the sense of common belief vs.\ common justifiability. Under general updating rules, the critical path theorem fails with respect to common belief but can hold with respect to common justifiability, and failures of common belief (resp.\ justifiability) are needed to generate contagion of secure (resp.\ arbitrary strict Nash) equilibria.}  While players in our main model share a common set of priors, Section~\ref{sec:alphaMEU} shows that our results extend to settings where players share a single non-probabilistic or second-order prior; thus, the departures from probabilistic uncertainty we highlight do not rely on ambiguity being modeled via multiple priors. By maintaining the common-prior assumption, our results also apply to settings where, starting from a complete-information benchmark, a designer induces an ambiguous type space through contracts or information, as in the unique-implementation problem in Section~\ref{sec:contracts}. By contrast, in line with the agreeing-to-disagree logic \citep{aumann1976}, inducing a non-common-prior type space would require the designer to manipulate agents' prior beliefs themselves, rather than merely provide them with different information.

 The unique-implementation application relates to a recent literature on the endogenous design of ambiguity. Several papers study how a designer can benefit from inducing ambiguity about agents' own payoffs via ambiguous information, contracts, or mechanisms \citep[e.g.,][]{tillio2016, kellner2018, beauchene2019, tang2021, dutting2024, cheng2024}. In a multi-agent mechanism design setting, \cite{bose2014} show that inducing first-order ambiguity (about other agents' payoff types) is a useful tool for partial implementation. In contrast, in our contracting setting, we highlight the channel of higher-order ambiguity (about other agents' beliefs) and its usefulness for unique implementation.

Finally, how to model belief updating under ambiguity is a key methodological question \citep[e.g.,][]{machina2014}. The literature has proposed and axiomatized many updating rules \citep[e.g.,][]{pires2002, gilboa1993, hanany2007, epstein2007, gul2021, cheng2019}. Some recent work compares different updating rules, for example, in terms of their performance in single-agent learning problems \citep{FII2025} or their evolutionary stability \citep{SS2024}. Our analysis in this paper accommodates a large class of updating rules and sheds light on their implications for higher-order beliefs. In particular, we highlight that less conservative alternatives to full Bayesian updating rules, such as maximum-likelihood updating and its variants, can magnify the higher-order belief effects of prior ambiguity, leading to more drastic strategic implications. Another well-known issue is that belief updating under ambiguity can give rise to dynamic inconsistency \citep[e.g.,][]{siniscalchi2009}. While dynamic inconsistency is an important force in many applications of ambiguity to strategic settings, Section~\ref{sec:dyn-inconsis} shows that our main insights do not hinge on this channel by considering a dynamically consistent equilibrium formulation \`a la \cite{hanany2020}.

\enlargethispage{\baselineskip}

\section{Setting}\label{sec:model}

\subsection{Type Space and Belief Updating}\label{sec:type}

Fix a finite set $I$ of at least two agents. Agents' higher-order ambiguity is captured by an \textit{\textbf{(ambiguous)  type space}} $(T = \prod_{i \in I} T_i, {\cal P})$, which specifies a countable set $T_i$ of types for each agent $i \in I$ and a compact set ${\cal P}\subseteq \Delta(T)$ of common prior distributions over type profiles.\footnote{For any topological space $X$, $\Delta(X)$ denotes the set of Borel probability measures on $X$, endowed with the induced topology of weak convergence. We endow $T_i$ with the discrete topology.}
 We assume that $P(\{t_i\}\times T_{-i})>0$ for all $t_i\in T_i$ and $P\in {\cal P}$, where $T_{-i}=\prod_{j\not=i}T_j$. The type space is unambiguous if ${\cal P} = \{ P \}$ is a singleton.

Agents know their own type, so the ambiguity of agent $i$ of type $t_i$ is captured by a set of posterior beliefs $\cP_{t_i}\subseteq\Delta(T_{-i})$ over opponent types. The literature has proposed many belief-updating rules under ambiguity. Perhaps the most commonly studied rule is \textit{\textbf{full Bayesian updating}}: For each prior $P \in \cP$, type $t_i$ forms the Bayesian update $P(\cdot \mid t_i) \in \Delta (T_{-i})$ (i.e., $P(t_{-i} \mid t_i ) = \frac{P(t_i, t_{-i})}{P(\{t_i\}\times T_{-i})}$ for all $t_{-i} \in T_{-i}$). Her set of posteriors is then
\[
\cP^{\rm FB}_{t_i}=\{P(\cdot \mid t_i): P \in \cP \}.
\]
This updating rule is quite conservative, as the agent does not use the information that her own type is $t_i$ to eliminate any priors.\footnote{See, e.g., \cite{seidenfeld1993} and \cite{shishkin2023} for conceptual and experimental critiques of full Bayesian updating.} 
Motivated by this, the literature has developed less conservative refinements of full Bayesian updating. For instance, under \textit{\textbf{maximum-likelihood updating}},
\[
\cP^{\rm ML}_{t_i}=\{P(\cdot \mid t_i):  P \in \argmax_{P' \in \cP} P'(\{t_i \} \times T_{-i}) \}.
\]
That is, after observing $t_i$, the agent restricts attention to the priors $P \in \cP$ that predicted $t_i$ with the highest probability and only forms the Bayesian updates of these priors. 

To capture these and many other commonly studied updating rules, we allow for any updating rule that satisfies the following assumption:

\begin{assume}\label{asp:ML}
For every $(T, \cP)$, $i \in I$, and $t_i \in T_i$, the set of posteriors $\cP_{t_i}$ is non-empty, closed, and satisfies:
\begin{enumerate}
\item[(i).] $\cP_{t_i} \subseteq \cP^{\rm FB}_{t_i}$,
\item[(ii).] for any $P, P'\in\cP$ with ${\rm  marg}_{T_i}P={\rm  marg}_{T_i}P'$,  
$$P(\cdot \mid t_i)\in \cP_{t_i} \iff P'(\cdot \mid t_i)\in \cP_{t_i}.$$
\end{enumerate}
\end{assume}

Condition (i) requires agents' updating rules to (weakly) refine full Bayesian updating. Condition (ii) imposes an own-information invariance requirement on this refinement: If $P$ and $P'$ have the same marginals over $T_i$, then they generate the same distribution over agent $i$'s private information, so $i$'s realized type $t_i$ cannot serve as a basis for eliminating one prior but not the other.\footnote{All our results remain valid if condition~(ii) is relaxed to the requirement that $\mathcal P_{t_i}$ contain a nonempty subset
$\widetilde{\mathcal P}_{t_i}$ satisfying condition~(ii).} Assumption~\ref{asp:ML} nests full Bayesian and maximum-likelihood updating, as well as natural hybrids whose sets of posteriors are in between $\cP^{\rm FB}_{t_i}$ and $\cP^{\rm ML}_{t_i}$ \citep[e.g.,][]{epstein2007, cheng2019}. It also nests some updating rules, such as proxy updating \citep{gul2021}, whose sets of posteriors need not be supersets of $\cP^{\rm ML}_{t_i}$.\footnote{Formally, proxy updating posits a proxy belief $\hat P_i \in \cP$ for each $i \in I$, and the set of posteriors for each $t_i$ is given by $\cP_{t_i}=\{P(\cdot \mid t_i): P \in \cP \text{ and } {\rm marg}_{T_i}P={\rm marg}_{T_i}\hat P_i \}$.} We allow updating rules to vary across agents and types.

\subsection{Common Belief}\label{sec:CB}

We will be interested in the impact of higher-order ambiguity on strategic interactions. \cite{monderer1989} introduced common $p$-belief as a natural quantification of the degree of higher-order uncertainty in strategic settings. We will make use of the following extension to ambiguous type spaces due to \cite{ganguli2007}.\footnote{\cite{ganguli2007} used this extension of common $p$-belief to generalize the agreeing-to-disagree result and the no-trade theorem to settings with ambiguity.}
Fix a type space $(T, \cP)$ and a belief-updating rule. Consider any $p \in [0, 1)$ and product event $E = \prod_{i \in I} E_i \subseteq T$. Define the event that $E$ is \textit{\textbf{individually ${p}$-believed}} by
 \begin{equation}\label{eq:Bp}
 B^{p}(E)=\prod_{i\in I}B^{p}_i (E), \quad \text{ where } \quad B^{p}_i(E)=\{t_i \in E_i : P_i(E_{-i}) \geq p \;  \forall P_i \in \cP_{t_i}\}.
\end{equation}
Define the event that $E$ is \textit{\textbf{commonly ${p}$-believed}} by
\[
C^{p}(E)=\bigcap_{k\in\mathbb N} (B^{p})^k(E),
\]
where $(B^{p})^0(E) = E$ and $(B^{p})^k(E)=B^{p}((B^{p})^{k-1}(E))$ for all $k\geq 1$. Section~\ref{sec:ML} also introduces a novel relaxation of common $p$-belief that will play an important role in contrasting the strategic impact of different belief-updating rules.

 \subsection{Ambiguous Incomplete-Information Games}\label{sec:game}

We model strategic interactions under ambiguity by an ambiguous \textbf{\textit{incomplete-information game}}: This consists of a type space $(T, \cP)$, along with a finite action set $A_i$ and bounded payoff function $v_i: A\times T\to\mathbb R$ for each agent $i \in I$, where $A=\prod_{j\in I}A_j$. As usual, each $v_i (\cdot, t)$ is extended linearly to mixed action profiles in $\prod_{j\in I}\Delta(A_j)$. A strategy for agent $i$ is a mapping $\sigma_i: T_i\to \Delta(A_i)$. We assume maxmin expected utility preferences \`a la \cite{GS1989} (Section~\ref{sec:alphaMEU} considers more general ambiguity preferences). That is, given any updating rule satisfying Assumption~\ref{asp:ML}, type $t_i$'s payoff to playing action $a_i$ against opponent strategies $\sigma_{-i} = (\sigma_j)_{j \neq i}$ is\footnote{The minimum over $P_i \in \cP_{t_i}$ is well-defined as $\cP_{t_i}$ is compact: By Assumption~\ref{asp:ML}, $\cP_{t_i}$ is a closed subset of $\cP^{\rm FB}_{t_i}$ (which is compact by compactness of $\cP$).}
\[ V_{i} (a_i, \sigma_{-i}, t_i) = \min_{P_i\in {\cal P}_{t_i}} \sum_{t_{-i}\in T_{-i}} P_i (t_{-i}) v_i(a_i, \sigma_{-i}(t_{-i}), t_i, t_{-i}).\]
We extend $i$'s payoffs linearly to mixed actions $\alpha_i \in \Delta(A_i)$, i.e., 
$$V_i (\alpha_i, \sigma_{-i}, t_i) = \sum_{a_i \in A_i} \alpha_i(a_i) V_{i} (a_i, \sigma_{-i}, t_i).$$ 
This implicitly assumes that a malevolent Nature chooses $P_i \in \cP_{t_i}$ after $i$'s pure action has been realized, so that $i$ cannot hedge against ambiguity by mixing. However, as Section~\ref{sec:mix} discusses, our main insights go through even if mixing is given hedging value by reversing this timing. An \textbf{\textit{equilibrium}} is a strategy profile $\sigma^*=(\sigma^*_i)_{i\in I}$ such that for each $i \in I$ and $t_i \in T_i$,
\[
\sigma^*_i(t_i)\in\argmax_{\alpha_i\in\Delta(A_i)} V_i (\alpha_i, \sigma^*_{-i}, t_i).
\]
Fixed-point arguments as in \cite{kajii2005} imply that an equilibrium exists. While this equilibrium formulation imposes interim optimality at each type $t_i$, Section~\ref{sec:dyn-inconsis} considers an alternative formulation based on ex-ante optimality.

\section{Main Results}\label{sec:main}

\subsection{Higher-Order Ambiguity vs.\ Probabilistic Uncertainty}\label{sec:CP}

Our first main result highlights a key difference between higher-order beliefs under ambiguity vs.\ probabilistic uncertainty. Since this result is purely epistemic, it does not rely on any of the assumptions in Section~\ref{sec:game} about agents' payoffs and ambiguity preferences in an underlying game. 

In unambiguous type spaces $(T, \{P\})$, a key insight due to \cite{kajii1997} is that if agents have high prior confidence in an event $E$, this limits how much higher-order uncertainty about $E$ there can be. Indeed, while \cite{rubinstein1989} highlighted that even if $P(E)$ is close to $1$, common ${p}$-belief of $E$ can fail for large ${p}$, \cites{kajii1997} critical path theorem shows that there is a lower bound on the probability of common ${p}$-belief at moderate ${p}$:

\medskip

\noindent {\bf Critical Path Theorem} \citep{kajii1997}{\bf.} {\it For every unambiguous type space $(T, \{P\})$, product event $E$, and ${p} \in (0, 1/|I|)$, 
\[
P(C^{p}(E)) \geq 1-(1-P(E))\frac{1-p}{1-|I|p}.
\]}

In particular, the critical path theorem implies that if $P(E) \approx 1$, then $P(C^{p}(E)) \approx 1$ for any $p < 1/|I|$. In contrast, our first main result shows a stark failure of the critical path theorem under ambiguity:

\begin{thm}\label{thm:CP-amb}
Fix any $\varepsilon \in (0, 1)$. There is a type space $(T, \cP)$ and product event $E$ such that $P(E) = P'(E) \geq 1-\varepsilon$ for all $P, P'\in \cP$ but $C^{ {\varepsilon} }(E)=\emptyset$ under all updating rules satisfying Assumption~\ref{asp:ML}.
\end{thm}

That is, under ambiguity, prior confidence imposes no limits on higher-order uncertainty: Even if an event $E$ has high probability under every prior belief $P\in \cP$, common ${\varepsilon}$-belief of $E$ can fail globally even for arbitrarily small ${\varepsilon}$. Observe that the event $E$ in Theorem~\ref{thm:CP-amb} is unambiguous, as all priors $P \in \cP$ assign the same probability to $E$. What drives the failure of the critical path theorem is that there is prior ambiguity about the probabilities of certain subevents of $E$.

More specifically, the proof of Theorem~\ref{thm:CP-amb} (Appendix~\ref{app:CP-amb}) adapts the construction in the illustrative example in Section~\ref{sec:example} to general updating rules and more than two agents. As in the example, we make use of type spaces with ordered types $t_i$ and two key features: (i) there is only small ambiguity about the marginal distributions of individual agents' types---indeed, the proof uses a type space where ${\rm marg}_{T_i} P$ is the same for all $P \in {\cal P}$, so marginals are fully unambiguous; but (ii) there is significant ambiguity about the correlation across agents' types, in particular, their relative rankings.\footnote{In contrast to the full ambiguity about relative rankings in the example, for any given $\varepsilon \in (0, 1)$, the prior probabilities that each agent's type is higher than her opponents' can be bounded away from zero and one.  Moreover, unlike in the example, $T_i$ and ${\cal P}$ in the proof of Theorem~\ref{thm:CP-amb} are finite.} The key insight behind Theorem~\ref{thm:CP-amb} is that these two features can give rise to events $E$ (e.g., the event that all agents' types are above some cutoff, as in the example) that have (unambiguous) high prior probability but nevertheless are subject to arbitrarily large higher-order uncertainty. In contrast, by the critical path theorem, this is impossible if the uncertainty about the relative rankings of types is probabilistic, as in the seminal email-game type space \`a la \cite{rubinstein1989}.

 As prior work has highlighted (e.g., in the literatures on global games and information design), type spaces with rank uncertainty are relevant in a variety of economic contexts, from currency attacks, where speculators are uncertain about how their private signals about fundamentals compare with those of others, to organizations, where employees may be uncertain about relative pay rankings. While such rank uncertainty is typically modeled as probabilistic, Theorem~\ref{thm:CP-amb} suggests that the higher-order belief implications can be very different if it is instead ambiguous.

It is worth noting that the type space $(T, {\cal P})$ and event $E$ in Theorem~\ref{thm:CP-amb} can be chosen uniformly across all updating rules satisfying Assumption~\ref{asp:ML}. Moreover, as we discuss in Section~\ref{sec:alphaMEU}, the fact that Theorem~\ref{thm:CP-amb} uses a multiple-prior formulation of ambiguity is not important. Analogs of Theorem~\ref{thm:CP-amb} can be derived when ambiguity is instead modeled via a common single second-order or non-additive prior.

 \subsection{Secure Actions and Contagion}\label{sec:strategic}

Building on Theorem~\ref{thm:CP-amb}, we now explore how ambiguity magnifies the impact of small amounts of payoff uncertainty in games.

To capture games with small amounts of payoff uncertainty, we fix a complete-information game $G = (A, u)$, where there is common certainty that $i$'s payoff function is $u_i: A\to\mathbb R$. We then consider \textbf{\textit{(ambiguous) $\varepsilon$-elaborations}} of $G$, i.e., incomplete-information games $(T, {\cal P}, (A_i, v_i)_{i\in I})$ as in Section~\ref{sec:game} such that
\begin{equation*}\label{eq:elaboration}
P\left(T^*\right)\geq 1-\varepsilon, \, \forall P\in{\cal P}.
\end{equation*} 
Here, $T^* = \prod_{i \in I} T^*_i$, where
\[
T^*_i =\left\{t_i\in T_i:  v_i(\cdot, (t_i, t_{-i}))=u_i(\cdot) \forall t_{-i} \in T_{-i}  \right\}
\]
denotes the set of $i$'s \textbf{\textit{normal types}}, i.e., types whose payoff functions coincide with the complete-information payoffs $u_i$. When $|\cP| = 1$, this corresponds to \cites{kajii1997} notion of unambiguous $\varepsilon$-elaborations. 

\begin{rem}\label{rem:small-ambig} 
In an $\varepsilon$-elaboration, payoff uncertainty is small in two senses: First, players have high confidence that payoff functions are given by the profile $u = (u_i)_{i \in I}$, as all priors assign probability at least $1 - \varepsilon$ to this event. Second, even though we allow for ambiguous type spaces, the first point also implies that there is only \emph{small prior ambiguity about game payoffs}: Indeed, for each $P \in \mathcal{P}$, consider the induced probability measure $\mu_P \in \Delta(\mathbb{R}^{A \times I})$ over game payoffs.\footnote{That is, $\mu_P (B) = P( \{ t \in T : (v_i (\cdot, t))_{i \in I} \in B\})$ for all measurable $B \subseteq \mathbb{R}^{A \times I}$.} Then the fact that $\mu_P (\{u\}) \geq 1 - \varepsilon$ for all $P\in \mathcal{P}$ implies that the total variation distance $d_{\rm TV} (\mu_P, \mu_{P'}) $ is at most $\varepsilon$ for any $P, P' \in \mathcal{P}$. In contrast, it is important for our results that $\varepsilon$-elaborations do not limit the distance between priors in $\cP$ when viewed as distributions over the entire type space; that is, there can be \emph{large prior ambiguity about payoff-irrelevant events}.\footnote{Indeed, Appendix~\ref{app:payoff-irrelevant} formalizes a sense in which, if $(T, \mathcal{P})$ features small prior ambiguity about all events in $T$, then it satisfies an approximate version of the critical path theorem, and as a result, $(T, \mathcal{P})$ does not support the contagion in Theorem~\ref{prop:contagious} below.}\finex
 \end{rem}
 
 To highlight the additional strategic implications of small payoff ambiguity relative to the case of small probabilistic uncertainty, we focus on the classic question of contagion \citep[][]{rubinstein1989}. Call a Nash equilibrium $a^*$ of $G$ \textbf{\textit{contagious}} if, 
 for every $\varepsilon>0$, there is an $\varepsilon$-elaboration of $G$ whose unique equilibrium $\sigma^*$ satisfies $\sigma^*_i( a^*_i|t_i)=1$ for all $i\in I$, $t_i\in T_i$. That is, even if game $G$ admits multiple equilibria, introducing arbitrarily small payoff uncertainty can select $a^*$ as the unique equilibrium.\footnote{Our results remain valid even if we strengthen the requirement that $a^*$ is played with probability 1 in the unique equilibrium of the elaboration to the requirement that $a^*$ is uniquely obtained via iterated elimination of dominated strategies.}  As we will see in Section~\ref{sec:implication}, understanding how ambiguity affects contagion has implications for two dual questions: First, which equilibrium predictions are robust when a cautious analyst perceives small exogenous ambiguity about players' payoffs (Section~\ref{sec:robustness}); second, which target actions a principal can uniquely implement when he can endogenously design such payoff ambiguity (Section~\ref{sec:contracts}).
 
Our second main result shows that ambiguity introduces a new force for contagion---security:
\begin{defn}
Action profile $a^*$ in $G$ is \textbf{\textit{secure}} if for every player $i$,
\begin{equation}\label{eq:gensafe}
\min_{a_{-i}\in A_{-i}} u_i(a^*_i, a_{-i}) >  \max_{a_i \neq a^*_i } u_i(a_i, a^*_{-i}).
\end{equation}
\end{defn}
That is, each player $i$'s worst-case payoff when she plays action $a^*_i$ and other agents play an arbitrary $a_{-i}$ is better than the best payoff she can achieve by deviating to any $a_i \neq a^*_i$ when other players stick with $a^*_{-i}$. 
Thus, security strengthens the strict Nash equilibrium requirement that $u_i(a^*_i, a^*_{-i}) >  \max_{a_i \neq a^*_i } u_i(a_i, a^*_{-i})$ by introducing worst-case reasoning about opponents' behavior into the left-hand side, i.e., into each player $i$'s calculation of her payoffs to playing $a^*_i$. Many important games admit secure equilibria. One example are strict Nash equilibria $a^*$ where each action $a^*_i$ is \textbf{\textit{safe}}, i.e., $i$'s payoff $u_i(a^*_i, a_{-i})$ is constant across opponent actions $a_{-i}$. Safe actions are natural in economic settings where there is an option of non-participation, as in the introductory example. However, secure equilibria also arise in important coordination games that do not feature safe actions:

\begin{ex}[Security without safety]\label{ex:secure}
Suppose $A = X^I$ and $u_i (a) = b(a) \mathbf{1}_{\{a_i = a_j \forall j\}} - c(a_i)$, where $b(a), c(a_i) \geq 0$ for all $a \in A$. That is, players receive a benefit if they all play the same action (e.g., adopt the same technology), while each player unilaterally incurs some cost for each action choice. Then coordinating on the least costly action is secure, i.e., playing $a^*$ with $a^*_i = x^*$ for all $i$ where $\{x^* \} = \argmin_{x \in X} c(x)$.  \finex\end{ex}

\begin{thm}\label{prop:contagious} 
Fix any complete-information game $G$. Under any updating rule satisfying Assumption~\ref{asp:ML}, any secure action profile $a^*$ in $G$ is contagious.
\end{thm}

Theorem~\ref{prop:contagious} implies that small payoff ambiguity has more extreme strategic implications than small probabilistic payoff uncertainty. 
To understand the gap, consider the introductory example, where $A_i = \{0, 1\}$ and $u_i(a)= \mathbf{1}_{a_i = a_{-i} =1} -ca_i$ for some $c \in (0, 1)$. There, $\mathbf{0}$ is contagious with respect to probabilistic $\varepsilon$-elaborations if and only if $c \geq \frac{1}{2}$, i.e., if and only if the other strict Nash equilibrium $\mathbf{1}$ is not risk-dominant. In contrast, since action $0$ is safe, Theorem~\ref{prop:contagious} implies that under ambiguous elaborations, $\mathbf{0}$ is contagious for \emph{all} $c > 0$ no matter how small, i.e., even if action $0$ is almost dominated by $1$. Likewise, while in general games $G$ no if-and-only-if characterization of contagion under probabilistic uncertainty is known, \cites{kajii1997} analysis implies that a necessary condition for $a^*$ to be contagious under probabilistic uncertainty is that no other equilibrium $\hat a$ is \textbf{\textit{$p$-dominant}} for some $p < 1/{|I|}$ (i.e., $\hat{a}_i$ is a strict best-response under $u_i$ to any opponent strategy that assigns probability at least $p$ to $\hat{a}_{-i}$); see Section~\ref{sec:robustness} for more discussion.\footnote{For generic binary-action supermodular games, 
\cite{oyama2020} provide an if-and-only-if condition for contagion under probabilistic uncertainty in terms of monotone potential maximization.}
In contrast, under ambiguity, Theorem~\ref{prop:contagious} implies that a secure $a^*$ is contagious even if another equilibrium $\hat a$ is $p$-dominant for arbitrarily small $p > 0$. Thus, vanishingly small payoff ambiguity can select a secure $a^*$ over an equilibrium in almost dominant strategies.

The following example illustrates how the security of $a^*$ allows one to exploit the failure of the critical path theorem in Theorem~\ref{thm:CP-amb} to construct elaborations where $a^*$ is the unique equilibrium. The proof of Theorem~\ref{prop:contagious} (Appendix~\ref{app:contagious}) adapts the ideas in the example to general games with more than two players.

\begin{ex}[Illustration of Theorem~\ref{prop:contagious}]\label{ex:contagious}
Consider the introductory example. To illustrate why $a^*={\bf 0}$ is contagious under general updating rules, consider $\varepsilon$-elaborations where $a^*_i = 0$ is dominant for non-normal types. Then, in any equilibrium $\sigma$, $T^* ({\bf 1}) :=\{t: \sigma({\bf 1}|t) > 0\}\subseteq T^*$. Moreover, since $a^*_i = 0$ is safe, a type $t_i$ is willing to choose action 1 in equilibrium only if \textbf{\emph{all posteriors}} of $t_i$ assign probability at least $p=c$ to $T^*({\bf 1})$; otherwise, the maxmin expected utility of playing action 1 is less than the payoff $0$ to choosing $0$.\footnote{More generally, if $a^*$ is secure, a type $t_i$ is willing to choose some action $a_i \neq a^*_i$ in equilibrium only if all of $t_i$'s posteriors assign sufficiently high probability to opponents deviating from $a^*_{-i}$.} Thus, $T^* ({\bf 1})\subseteq B^p(T^* ({\bf 1}))$ (i.e., $T^* ({\bf 1})$ is  $p$-evident), which (by standard arguments) implies $T^* ({\bf 1})\subseteq C^p (T^*)$. But by Theorem~\ref{thm:CP-amb}, the elaboration can be constructed in such a way that $C^p(T^*)=\emptyset$. Thus, $\sigma ({\bf 0} \mid t) = 1$ for all $t$. \finex
\end{ex}

\subsection{Role of Belief Updating}\label{sec:ML}

While Theorems~\ref{thm:CP-amb}--\ref{prop:contagious} apply under general updating rules, some updating rules give rise to results that are even more extreme departures from probabilistic uncertainty. Indeed, we now illustrate that updating rules that are less conservative (i.e., induce smaller sets of posteriors) can magnify the contagious effects of ambiguity. For expositional simplicity, we focus on contrasting the predictions under the most conservative updating rule satisfying Assumption~\ref{asp:ML}, full Bayesian updating, with those under the refinement of maximum-likelihood updating. Appendix~\ref{app:CPJ-FB} provides intermediate results for a class of updating rules that are in between these two updating rules.

To contrast the higher-order belief effects of ambiguity across different updating rules, we introduce the following relaxation of ${p}$-belief. Define the event that $E$ is \textit{\textbf{individually ${p}$-justifiable}} by
\[
\overline B^{p}(E)=\prod_{i\in I}\overline B^{p}_i(E), \quad \text{ where }  \overline B^{p}_i(E)=\{t_i \in E_{i}: \exists P_i \in \cP_{t_i} \text{ s.t. }  P_i(E_{-i}) \geq p \}.
\]
Define the event that $E$ is \textit{\textbf{commonly ${p}$-justifiable}} by
\[
\overline C^{p}(E)=\bigcap_{k\in\mathbb N} (\overline B^{p})^k(E),
\]
where $(\overline{B}^{p})^0(E) = E$ and $(\overline{B}^{p})^k(E)=\overline{B}^{p}((\overline{B}^{p})^{k-1}(E))$ for all $k\geq 1$. Note that $\overline B^{p}_i(E)$ replaces the universal quantifier over posteriors $P_i \in \cP_{t_i} $ in the $p$-belief operator $B^p_i (E)$ with an existential quantifier. In unambiguous type spaces, $B^{p}(E)= \overline B^{p}(E)$ and $C^{p}(E) = \overline C^{p}(E)$; more generally, $B^{p}(E)\subseteq \overline B^{p}(E)$, and hence $C^{p}(E)\subseteq \overline C^{p}(E)$ by standard monotonicity arguments. Thus, common $p$-justifiability is in general a weaker requirement than common $p$-belief. 

\begin{prop}\label{prop:CP-ML} $ $
\begin{enumerate}

\item Under full Bayesian updating, the critical path theorem holds under common justifiability:\footnote{We use ``under full Bayesian updating'' to mean that for every $(T, \cP)$, $i \in I$, and $t_i \in T_i$, we have $\cP_{t_i} = \cP^{\rm FB}_{t_i}$, with the analogous meaning for ``under maximum-likelihood updating''.} For every type space $(T, \cP)$, product event $E$, $P \in \cP$, and $p \in (0, 1/|I|)$, we have
$P(\overline C^{p}(E)) \geq 1-(1-P(E))\frac{1-p}{1-|I|p}$.

\item Under maximum-likelihood updating, Theorem~\ref{thm:CP-amb} extends to common justifiability: For any $\varepsilon \in (0, 1)$, there is a type space $(T, \cP)$ and a product event $E$ such that $P(E)= P'(E) \geq 1-\varepsilon$ for all $P, P'\in \cP$ but $\overline C^{ {\varepsilon}} (E)=\emptyset$.

\end{enumerate}

\end{prop}

Thus, under full Bayesian updating, high prior confidence in an event $E$ can coexist with failures of common $\varepsilon$-belief of $E$ for arbitrarily small $\varepsilon$, but it at least guarantees common $p$-justifiability of $E$ for moderate $p$. In contrast, maximum-likelihood updating exacerbates the higher-order belief effects of ambiguity: No amount of prior confidence in $E$ is enough to even rule out failures of common $\varepsilon$-justifiability for arbitrarily small $\varepsilon$. 

\begin{ex}[Illustration of Proposition~\ref{prop:CP-ML}]\label{ex:CP-ML} Consider the introductory example. Under maximum-likelihood updating, every normal type $t_1 \geq 1$ of player 1 has the unique posterior $P^1 (\cdot \mid t_1)$, as $P^1(\{t_1\}\times T_{2})>P^\lambda(\{t_1\} \times T_{2})$ for any $\lambda\not=1$. Similarly, for every normal type $t_2 \geq 1$ of player 2, ${\cal P}_{t_2}^{\rm ML} = \{P^0 (\cdot \mid t_2) \}$.\footnote{However, crazy types $t_1=0$ (resp. $t_2=0$) use $P^0$ (resp. $P^1$) to form posteriors.} Thus, each normal type $t_i $'s unique posterior assigns probability one to the event that $t_{-i}=t_{i}-1$. As a result, the set of normal type profiles $E=\{1,2,\ldots\}^2$ is not individually $\varepsilon$-justifiable at $t_i=1$, and inductively, $E$ is not commonly $\varepsilon$-justifiable at \emph{any} type profile.\finex
\end{ex}

Under maximum-likelihood updating, the stronger failure of the critical path theorem in Proposition~\ref{prop:CP-ML} also gives rise to an extreme form of contagion that does not rely on security. (For the intermediate class of updating rules in Appendix~\ref{app:CPJ-FB}, we instead provide a sufficient condition for contagion that retains a weaker form of security.)

\begin{prop}\label{prop:contagious-ML}
Fix any complete-information game $G$. Under maximum-likelihood updating, any strict Nash equilibrium $a^*$ of $G$ is contagious.
\end{prop}

The following example illustrates the connection between failures of common justifiability and the contagion result for general strict Nash equilibria. Beyond the question of contagion, the contrast between Examples~\ref{ex:contagious} and \ref{ex:contagious2} (and the proofs of Theorem~\ref{prop:contagious} and Proposition~\ref{prop:contagious-ML}) highlights a broader methodological point: Under ambiguity, different belief operators---$p$-belief vs.\ $p$-justifiability---are needed to capture players' coordination incentives for secure vs.\ non-secure actions.

\begin{ex}[Illustration of Proposition~\ref{prop:contagious-ML}]\label{ex:contagious2} Consider the non-secure action profile ${\bf 1}$ in the introductory example. To see why this is contagious under maximum-likelihood updating, consider $\varepsilon$-elaborations where $a^*_i = 1$ is dominant for non-normal types, so $T^* ({\bf 0}) :=\{t: \sigma({\bf 0}|t) > 0\}\subseteq T^*$ in any equilibrium $\sigma$. Now, $t_i$ is willing to choose action 0 in equilibrium only if \textbf{\emph{some posterior}} of $t_i$ assigns probability at least $p=1-c$ to $T^*({\bf 0})$; otherwise, the maxmin expected utility of playing action 1 is greater than the payoff $0$ to choosing $0$. Thus, in contrast with the requirement that $T^*({\bf 1}) \subseteq  B^p(T^*({\bf 1}))$ in Example~\ref{ex:contagious}, we now need $T^*({\bf 0})\subseteq \overline B^p(T^*({\bf 0}))$, and hence  $T^* ({\bf 0}) \subseteq \overline C^p (T^*)$. But under maximum-likelihood updating, Proposition~\ref{prop:CP-ML} shows we can construct an elaboration where even $\overline C^p(T^*)=\emptyset$. Then, $\sigma ({\bf 1} \mid t) = 1$ for all $t$. \finex
\end{ex}

\section{Implications}\label{sec:implication}

We now highlight the dual implications of our main results for robust equilibrium predictions and unique implementation in games. Throughout this section, for each type space $(T, \cP)$, we fix some arbitrary updating rules $(\cP_{t_i})$ satisfying Assumption~\ref{asp:ML}.

\subsection{Equilibrium (Non-)Robustness}\label{sec:robustness}

Consider the question of equilibrium robustness studied by \cite{kajii1997} and a large follow-up literature.\footnote{See, e.g., \cite{ui2001, morris2005, oyama2009, chassang2011, oyama2020, pei2025}, and the survey by \cite{KMsurvey2, KMsurvey}. More broadly, this question also underlies the large literature on robust mechanism design \citep[for a survey, see][]{bergemann2013}.} An analyst views a complete-information game $G = (A,u)$ as only an approximation of the true environment, allowing for the possibility that players face small exogenous payoff uncertainty. Crucially, the analyst is cautious with respect to such payoff uncertainty and takes a robust approach to equilibrium predictions: for an equilibrium of $G$ to be robust, he requires that it remain a valid prediction under {\it every} $\varepsilon$-elaboration with sufficiently small $\varepsilon$.

Formally, extending \cites{kajii1997} definition to allow for payoff ambiguity, call an equilibrium $\hat a$ of $G$ \textbf{\textit{robust}} if for every $\gamma < 1$, there is $\overline\varepsilon>0$ such that every $\varepsilon$-elaboration of $G$ with $\varepsilon\leq\overline\varepsilon$ admits an equilibrium $\sigma$ with $\sum_{t\in T}P(t)\sigma(\hat a|t)\geq \gamma$ for all $P\in \cP$; that is, in all close enough elaborations of $G$, $\hat a$ is played with high probability in some equilibrium.\footnote{Note that since we are fixing some updating rules $(\cP_{t_i})$ for each type space $(T, {\cal P})$, this robustness notion is weaker than one that would require robustness under all possible updating rules satisfying Assumption~\ref{asp:ML}. This makes the implications we state below for \emph{non}-robustness stronger.} For $\hat a$ to be robust, it is necessary that no other equilibrium $a^*$ of $G$ is contagious. Thus, Theorem~\ref{prop:contagious} immediately implies the following:

\begin{cor}\label{cor:non-robust}
Fix any complete-information game $G$.  An equilibrium $\hat a$ of $G$ is not robust if there is some secure action profile $a^*\not=\hat a$.  
\end{cor}

The literature has developed various sufficient conditions for an equilibrium $\hat a$ to be robust to probabilistic payoff uncertainty. In particular, \cite{kajii1997} show that this is the case if $\hat a$ is $p$-dominant for $p < 1/|I|$, which reduces to the requirement that $\hat a$ is risk-dominant in (symmetric) $2\times2$ coordination games; follow-up work has introduced other sufficient conditions that also reduce to risk-dominance in $2\times2$ coordination games. The key implication of Corollary~\ref{cor:non-robust} is that none of these generalizations of risk-dominance is sufficient to guarantee robustness to payoff ambiguity: As we saw in the introductory $2 \times 2$ game, even if $\hat a$ is $p$-dominant for $p \approx 0$, this does not rule out the presence of other secure equilibria. 

At the same time, our analysis also yields sufficient conditions for robustness to payoff ambiguity by appropriately combining $p$-dominance with security. Call $\hat a \in A$ \textit{\textbf{$p$-secure}} if for all $i$,
\begin{equation}\label{eq:p-secure}
\min_{a_{-i} \in A_{-i}} u_i(\hat a_i, a_{-i}) > \max_{a_i\neq \hat a_i} u_i(a_i, \alpha_{-i}), \quad  \forall \alpha_{-i}\in\Delta(A_{-i}) \text{ with  } \alpha_{-i} (\hat a_{-i})\geq p.
\end{equation}
Note that replacing the left-hand side of (\ref{eq:p-secure}) with $u_i(\hat a_i, \alpha_{-i})$ yields $p$-dominance. Thus, analogous to how security strengthens strict Nash, $p$-security strengthens $p$-dominance by incorporating worst-case reasoning about opponents' behavior into each $i$'s calculation of her payoffs to playing $\hat a_i$. Likewise, (\ref{eq:p-secure}) strengthens the security requirement $\min_{a_{-i} \in A_{-i}} u_i(\hat a_i, a_{-i}) > \max_{a_i\neq \hat a_i} u_i(a_i, \hat a_{-i})$ by incorporating $p$-dominance reasoning into the right-hand side. Any $p$-dominant $\hat a$ where $\hat a_i$ is safe for all $i$ is $p$-secure.

Under full Bayesian updating, Proposition~\ref{prop:CP-ML} can be used to show that $p$-security for $p < 1/|I|$ is sufficient for robustness. At the same time, since less conservative updating rules can exacerbate contagion, the corresponding conditions for robustness are in general more demanding. Indeed, under the extreme case of maximum-likelihood updating, Proposition~\ref{prop:contagious-ML} implies that no equilibrium is robust when there are multiple strict Nash equilibria, so no degree of $p$-security with $p > 0$ ensures robustness. More generally, Appendix~\ref{app:CPJ-FB} considers a parametrized class of updating rules in between full Bayesian and maximum-likelihood: Now, $p$-secure equilibria for $p < \overline p$ are robust, where $\overline p$ is smaller the less conservative the updating rule.

\begin{cor}\label{cor:FB/ML} \
\begin{enumerate}
\item  Under full Bayesian updating, any $p$-secure equilibrium $\hat a$ of $G$ with $p<1/|I|$ is robust.
\item Under maximum-likelihood updating, no game $G$ with multiple strict Nash equilibria admits a robust equilibrium.
\end{enumerate}

\end{cor}

\subsection{Unique Implementation in Games}\label{sec:contracts}

Consider a classic problem of contracting with externalities \citep[for a survey, see][]{halac2025}. Given some complete-information game $G = (A, u)$ describing players' interdependent incentives, a principal seeks to design payments to each player to implement some target action profile $a^* \in A$ as an equilibrium; crucially, to avoid coordination on another, less desirable equilibrium, the principal requires the equilibrium to be unique. Formally, the principal can offer each agent $i$ a bilateral contract that pays $w_i \geq 0$ if and only if $i$ plays action $a^*_i$, and the principal's goal is to choose the cheapest payment scheme that induces $a^*$ as the unique equilibrium.

Earlier work studied this unique implementation problem when the principal is restricted to deterministic payments $w_i$ \citep[e.g.,][]{segal2003}. More recently, \cite{halac2021} highlighted that the principal can do better by introducing probabilistic higher-order uncertainty via stochastic payments $w_i$.\footnote{We focus on contractible actions \citep[as in][]{segal2003, sakovics2012}, but our analysis can be extended to settings with moral hazard \citep[as in][]{halac2021}.} Moving beyond this, we allow the principal to design payment schemes that involve higher-order ambiguity, in the sense that each agent $i$ knows her own payment but faces ambiguity about others' payments, their beliefs about $i$'s payment, etc. Such ambiguity is natural in many contracting problems with externalities; for example, in organizations managers may provide employees with only vague information about coworkers' wages, and in the context of technology adoption, firms may offer personalized subsidies to new adopters, creating ambiguity among consumers about each other's offers. Based on the previous analysis, we show that ambiguous payments can lead to further significant reductions in the principal's unique implementation cost.

 {\bf Principal's problem.} The principal chooses an \textbf{\textit{ambiguous payment scheme}} $\cC = (T, \cP, w)$: Here, $(T, \cP)$ is an ambiguous type space, and $w = (w_i)_{i\in I}$ specifies a bounded payment $w_i :  T_i\to\mathbb R_+$ for each agent $i$ as a function of her type $t_i$, which is paid if and only if $i$ plays $a^*_i$. We identify each payment scheme $\cC = (T, \cP, w)$ 
with the induced ambiguous incomplete-information game with payoff functions
\[
v_i(a, t)= u_i(a) + w_i (t_i) {\bf 1}_{a_i=a^*_i}.
\]
Denote by $\Sigma(\cC)$ the set of equilibria of $\cC$. Slightly abusing notation, we denote by $a^*$ the strategy profile $\sigma$ of $\cC$ such that $\sigma (a^* \mid t) = 1$ for all $t \in T$.

The principal seeks to choose the cheapest payment scheme $\cC$ that implements $a^*$ as the unique equilibrium of $\cC$: The \textbf{\textit{unique implementation cost}} $W^{a^*}$ is given by
\begin{equation}\label{eq:implementation}
W^{a^*} =\inf_{\cC = (T, \cP, w) } \min_{P\in \cP} \sum_{t\in T} P(t)\sum_{i\in I}w_i(t_i) \quad  \text{ subject to }
\end{equation} 
\begin{equation}\tag{unique implementation}\label{eq:full}
\{a^*\}=\Sigma({\cal C}), 
\end{equation}
\begin{equation}\tag{consistency}\label{eq:consistency}
\sum_{t\in T} P(t)\sum_{i\in I}w_i(t_i)=\sum_{t\in T} P'(t)\sum_{i\in I}w_i(t_i), \;\; \forall P,P'\in \cP.
\end{equation}

Deterministic payment schemes correspond to $|T| = |\cP| = 1$; more generally, unambiguous payment schemes correspond to $|T| \geq 1$ but $|\cP| = 1$, i.e., allow for probabilistic higher-order uncertainty. Let $W^{a^*}_{\rm prob}$ denote the unique implementation cost when the principal is restricted to unambiguous payment schemes.

In allowing for ambiguous payment schemes, we impose a consistency condition as in \cite{tillio2016} and \cite{dutting2024}, which requires the principal to be indifferent among all distributions $P \in {\cal P}$. The interpretation is that this is necessary for the principal to be able to credibly introduce ambiguity about which of several payment schemes he is using; otherwise, agents might be able to narrow down the set of possible schemes by reasoning about the principal's incentives.

The implementation cost in (\ref{eq:implementation}) considers a $P \in \cP$ that minimizes the principal's expected payment. An interpretation is that the principal himself knows which $P \in \cP$ he uses, but keeps the choice of $P \in \cP$ ambiguous for agents (e.g., by using vague language or by conditioning the choice of $P$ on an event, such as the composition of an Ellsberg urn, whose realization is known to the principal but ambiguous to agents). However, by the consistency condition, (\ref{eq:implementation}) could equivalently consider $P \in \cP$ that maximize the principal's expected payment, capturing a principal who commits to also not learning the realized $P \in \cP$ (e.g., by delegating its choice to an intermediary) and has maxmin expected utility preferences.

{\bf Unique vs.\ partial implementation.} To illustrate how ambiguous payment schemes can facilitate unique implementation, we contrast the unique implementation cost $W^{a^*}$ with the \textbf{\textit{partial implementation cost}} $W^{a^*}_{\rm partial}$, i.e., the value of (\ref{eq:implementation}) when the unique implementation condition $\{a^*\}=\Sigma({\cal C})$ is relaxed to the condition that $a^* \in \Sigma(\cC)$. For partial implementation, there is no benefit to using ambiguous payment schemes: $W^{a^*}_{\rm partial}$ is achieved by a deterministic payment scheme with payments $w^*_{i, {\rm partial}}=\max_{a_i\in A_i}u_i(a_i, a^*_{-i})-u_i(a^*)$ for each $i$ (see Lemma~\ref{lem:partial}).

Under partial implementation, the principal accepts some strategic uncertainty about which equilibrium of $\cC$ will be played. Consequently, partial implementation is in general strictly less costly than unique implementation. However, based on Theorem~\ref{prop:contagious}, we exhibit two settings in which ambiguous payment schemes bring the unique implementation cost all the way down to the partial implementation cost, i.e., allow the principal to fully eliminate strategic uncertainty for free.

{\bf Unique implementation of secure actions.} First, consider target actions $a^*$ that are sufficiently close to secure:
\begin{cor}\label{lem:FB bound}
We have $W^{a^*}=W^{a^*}_{\rm partial}$ if $a^*$ satisfies
\begin{equation}\label{eq:general-impl}
\min_{a_{-i}\in A_{-i}}u_i(a^*_i, a_{-i}) + w^*_{i, {\rm partial}}  \geq \max_{a_i\not=a^*_i}u_i(a_i, a^*_{-i}) \quad \text{ for all } i\in I.
\end{equation} 
\end{cor}
Condition (\ref{eq:general-impl}) requires that in the modified complete-information game where each agent receives the partial implementation subsidy $w^*_{i, {\rm partial}}$ for playing $a^*_i$, $a^*$ satisfies (a weak inequality version of) security.  Clearly, (\ref{eq:general-impl}) holds if $a^*$ is secure in the original game $G$ (in which case $w^*_{i, {\rm partial}} = 0$ as $a^*$ is already an equilibrium of $G$); (\ref{eq:general-impl}) also holds if each $a^*_i$ is safe, regardless of whether or not $a^*$ is an equilibrium of $G$.

Corollary~\ref{lem:FB bound} highlights that for actions that satisfy condition (\ref{eq:general-impl}), ambiguous payment schemes greatly facilitate unique implementation relative to unambiguous payment schemes. Indeed, even for secure actions, unique implementation using unambiguous payment schemes can be arbitrarily more costly than partial implementation:

\begin{ex}[Ambiguous vs.\ unambiguous payment schemes]\label{ex:contracts}
Consider an instance of the technology adoption game in Example~\ref{ex:secure}: $A_i = \{0, 1\}$ and $u_i (a_i, a_{-i}) = b(a_i)\mathbf{1}_{a_i = a_{-i}} - ca_i $ for all $i = 1, 2$, where $b(1) > b(0) > c > 0$. Since ${\bf 0}$ is secure, Corollary~\ref{lem:FB bound} implies that, under ambiguous payment schemes, $W^{\bf 0} = W^{\bf 0}_{\rm partial} = 0$ for all $b$ and $c$.
In contrast, under unambiguous payment schemes, Appendix~\ref{app:unambig-impl} shows that $W^{\bf 0}_{\rm prob} = \max\{b(1)-b(0)-2c, 0\}$. In particular, $W^{\bf 0}_{\rm prob} \to \infty$ as $b(1) - b(0) \to \infty$. \finex
\end{ex}

 The proof of Corollary~\ref{lem:FB bound} (Appendix~\ref{app:FB-contract}) builds on Theorem~\ref{prop:contagious} by designing an ambiguous payment scheme that, given condition (\ref{eq:general-impl}), creates contagion toward $a^*$: All normal types receive a payment $w^*_{i, {\rm partial}} + \delta$ with $\delta \approx 0$, whereas non-normal types receive a $w_i$ that is large enough to make $a^*_i$ a dominant action. Here, analogous to Theorems~\ref{thm:CP-amb}--\ref{prop:contagious}, players know their own payments but face ambiguity about their type's ranking relative to other players, which generates higher-order ambiguity about other players' (beliefs about) payments. Relative to Theorem~\ref{prop:contagious}, additional care is required to ensure that the consistency condition is satisfied. However, the amount of payoff ambiguity introduced by the principal is again small, in that the probability of normal types is arbitrarily close to $1$ under all $P \in \cP$.  Just as in Theorems~\ref{thm:CP-amb}--\ref{prop:contagious}, the same ambiguous payment scheme can be used for all updating rules satisfying Assumption~\ref{asp:ML}, so that the principal does not need to know players' updating rules.\footnote{That said, if the principal knows players' updating rules, then $W^{a^*}=W^{a^*}_{\rm partial}$ may hold for a larger class of target actions (e.g., for all $a^*$ under maximum-likelihood updating, by Proposition~\ref{prop:contagious-ML}).}

{\bf Unique implementation via insurance contracts.} Second, we exhibit a limited and natural departure from bilateral contracts---insurance contracts---that, when combined with ambiguity, reduces the unique implementation cost of \emph{any} target action to its partial implementation cost.

An insurance contract reimburses each player $i$ for any loss caused by other players' deviations from the target action profile $a^*$ provided $i$ plays $a^*_i$, and thus makes $a^*_i$ a safe action.\footnote{This idea bears a broad resemblance to some schemes used in practice. For example, FDIC deposit insurance discourages bank runs by making non-withdrawal safe for insured depositors.} Formally, it specifies payments $w_i \geq 0$ for each $i$ such that if $i$ plays $a^*_i$, she receives a total subsidy of $u_i (a^*) - u_i (a^*_i, a_{-i}) + w_i$ for any opponent action profile $a_{-i}$, and hence a constant payoff of $u_i (a^*) + w_i$ regardless of $a_{-i}$; if $i$ does not play $a^*_i$, $i$ receives no subsidy.\footnote{This formulation allows for negative off-path subsidies. However, the same result continues to hold if subsidies are replaced with $\max \{ u_i (a^*) - u_i (a^*_i, a_{-i}) + w_i, 0 \}$, ensuring limited liability.} Allowing for higher-order ambiguity about $(w_i)$, an \textit{\textbf{ambiguous insurance contract}} specifies a type space $(T, \cP)$ and bounded subsidies $w_i : T_i\to\mathbb R_+$ that induce the incomplete-information game with payoffs
\[
v_i(a, t)= 
\begin{cases}
u_i(a^*)+ w_i (t_i) &\text{ if } a_i=a^*_i \\
u_i(a) & \text{ otherwise. }
\end{cases}
\]
Denote by $\overline W^{a^*}$ the unique implementation cost of $a^*$ when in (\ref{eq:implementation}) the principal optimizes over ambiguous insurance contracts, and by $\overline W^{a^*}_{\rm partial}$ the corresponding partial implementation cost. These problems reduce to the original bilateral contracting setting with modified game payoffs $\overline u_i$ given by $\overline u_i(a)=u_i(a^*)$ if $a_i=a^*_i$ and $\overline u_i(a)= u_i(a)$ otherwise. Since each agent's target action $a^*_i$ is safe under this modified game, Corollary~\ref{lem:FB bound}  immediately implies the following:

\begin{cor}\label{cor:insurance} For every $a^*\in A$, we have $\overline W^{a^*}=\overline W^{a^*}_{\rm partial}=W^{a^*}_{\rm partial}$. 
\end{cor}

Thus, when combined with (even small amounts of) payoff ambiguity, insurance contracts are a highly effective tool for unique implementation. The combination with ambiguity is key, as the unique implementation cost $\overline W^{a^*}_{\rm prob}$ under unambiguous insurance contracts in general strictly exceeds $W^{a^*}_{\rm partial}$. For instance, in Example~\ref{ex:contracts}, $\overline W^{\bf 1}_{\rm prob}=  \max \{2c-2b(1)+b(0), 0\}$ (see Appendix~\ref{app:unambig-impl}), which can be arbitrarily more costly than $W^{\bf 1}_{\rm partial} = 0$.\footnote{As the literature notes, fully general non-bilateral contracts can equalize unique- and partial-implementation costs even with deterministic payments, but there are various difficulties with using such contracts \citep[e.g.,][]{segal2003, halac2025}. Insurance contracts mitigate some of these concerns: (i) while this broader class may condition $i$'s payment on the details of opponents' actions, which may be very costly to verify, insurance contracts only condition on $i$'s action and payoff loss $u_i(a^*)-u_i(a_i^*,a_{-i})$; (ii) while the aforementioned deterministic contracts make each $a^*_i$ strictly dominant via large payments if opponents deviate, insurance contracts merely make $a^*_i$ safe, which requires less knowledge of the full payoff matrix and makes them less vulnerable to side contracts among agents.}


\section{Discussion}\label{sec:discussion}

\subsection{More General Ambiguity Representations}\label{sec:alphaMEU}

Beyond maxmin expected utility, our main insights extend to more general classes of ambiguity-sensitive preferences. First, to capture players with a mix of ambiguity-averse and ambiguity-seeking tendencies, Online Appendix~\ref{app:alpha-MEU} considers $\alpha$-maxmin expected utility. We provide a condition for contagion that reduces to security under maxmin expected utility.  Under the opposite extreme of maxmax preferences, the condition implies that a strict Nash equilibrium $a^*$ is contagious if every action $a_i \neq a^*_i$ is safe. Thus, if agents are ambiguity-seeking, the presence of safe actions serves as a source of contagion for \emph{non-safe} equilibria. However, as in the maxmin case, this again means that almost-dominated actions can be contagious.

Second, while we have followed the canonical approach of representing ambiguity by a set of priors, our main results extend to other standard representations of ambiguity based on a single second-order prior or non-additive prior. In particular, in Online Appendix~\ref{app:smooth}--\ref{app:CEU}, we adapt ambiguous type spaces and elaborations to agents with smooth preferences \citep{klibanoff2005} with a common second-order prior $\mu \in \Delta(\Delta(T))$, as well as to agents with Choquet expected utility preferences \citep{schmeidler1989} and a common non-additive Choquet capacity $\nu$ over $T$. We derive analogs of the failure of the critical path theorem in Theorem~\ref{thm:CP-amb}. We also show that secure action profiles $a^*$ remain contagious under both these models.

\subsection{Hedging against Ambiguity through Mixing}\label{sec:mix}

The formulation of incomplete-information games in Section~\ref{sec:game} rules out hedging against ambiguity through mixing: When player $i$ mixes, Nature selects a
worst-case posterior $P_i \in \mathcal{P}_{t_i}$ after $i$'s pure action has been realized. An alternative formulation reverses this order, thereby giving mixing hedging value. That is, Nature selects a
posterior after observing $i$'s mixed action but before the corresponding pure action is realized, so that $i$'s interim payoff from a mixed
action $\alpha_i\in\Delta(A_i)$ is
\[
\widetilde V_i(\alpha_i,\sigma_{-i},t_i)
=
\min_{P_i\in\mathcal P_{t_i}}
\sum_{t_{-i}\in T_{-i}}
P_i(t_{-i})
v_i(\alpha_i,\sigma_{-i}(t_{-i}),t_i,t_{-i}).
\]

In general, it is well known that giving mixing hedging value can substantially change equilibrium predictions.\footnote{For
example, in \cite{dutting2024} the advantage of ambiguous contracts
over unambiguous contracts disappears when the agent can hedge by mixing. See also \cite{bose2014, tillio2016}.} However, our main insights go through under this alternative formulation. In particular, every strict Nash equilibrium in safe actions remains
contagious, and the implications in Section~\ref{sec:implication} continue to
hold whenever security is strengthened to safety. The reason is that the contagion argument for equilibria in safe actions $a^*_i$ does not require Nature to choose
different posteriors against different pure deviations from $a^*_i$. Rather, at each
inductive step, the construction yields a single posterior
$P_i\in\mathcal P_{t_i}$ for the type $t_i$ under consideration under which \emph{every} pure
$a_i\neq a_i^*$ has expected payoff strictly below the constant payoff
from the safe action $a_i^*$. As a result, this also implies that $t_i$'s worst-case expected payoff of every mixed action $\alpha_i\neq\delta_{a_i^*}$ is strictly lower than the payoff from $a^*_i$. Hence, $\delta_{a_i^*}$ is the unique surviving action for type $t_i$,
allowing the inductive argument to proceed as before.\footnote{However, this argument need not extend to secure $a^*$ where $a^*_i$ is not safe. The construction still yields a single posterior $P_i\in\mathcal P_{t_i}$ under which every pure $a_i\neq a_i^*$ has an expected payoff strictly below the \emph{worst-case} payoff of $a_i^*$. But when $a^*_i$ is not safe, this does not rule out that some mixed $\alpha_i\neq\delta_{a_i^*}$ with $a^*_i \in {\rm supp} \, \alpha_i$ yields a better worst-case payoff than $a^*_i$.}

\subsection{Dynamically Consistent Agents}\label{sec:dyn-inconsis}

As is well known, belief-updating under ambiguity gives rise to dynamic inconsistency: A strategy that maximizes a player's interim worst-case payoff based on her set of posteriors ${\cal P}_{t_i}$ at each type need not in general maximize her ex-ante worst-case payoff based on the set of priors $\mathcal{P}$. However, dynamic inconsistency is not what drives the departures from probabilistic uncertainty that we have highlighted.  

To shut down the dynamic inconsistency channel, we replace the equilibrium notion in Section~\ref{sec:game} based on interim optimality with one based on ex-ante optimality.\footnote{One interpretation of this solution concept is that players follow the dynamically consistent updating rule in \cite{hanany2007}. Posteriors under this updating rule are subsets of ${\cal P}^{\rm FB}_{t_i}$, but this rule violates consequentialism, in that ${\cal P}_{t_i}$ can depend on opponent strategies and game payoffs. In contrast, our interim-optimal analysis imposed consequentialism, because we fixed players' updating rules across opponent strategies and games.} Define
$i$'s ex-ante payoff to playing pure strategy $s_i: T_i\to A_i$ against opponents' mixed strategy profile $\sigma_{-i}$ by
\[
V_i(s_i, \sigma_{-i})=\min_{P\in \cP}\sum_{t\in T}P(t)v_i \left(s_i(t_i), \sigma_{-i}(t_{-i}), t_i, t_{-i} \right),
\]
and extend $V_i$ linearly to mixed strategies $\sigma_i $ of player $i$.\footnote{We embed mixed strategies $\sigma_i  \in (\Delta(A_i))^{T_i}$ into mixtures $\Delta(A_i^{T_i})$ over pure strategies in the standard way.} An \textbf{\textit{ex-ante equilibrium}} is a strategy profile $\sigma^*=(\sigma^*_i)_{i\in I}$ such that, for each $i \in I$, 
$\sigma^*_i\in\argmax_{\sigma_i}V_i(\sigma_i, \sigma^*_{-i})$. Thus, each player's contingent plan $\sigma^*_i $ is a worst-case best-response against $\sigma^*_{-i}$ based on the set of priors $\mathcal{P}$. Analogously to Section~\ref{sec:strategic}, for any complete-information game $G$, we call $a^* \in A$ \textbf{\textit{ex-ante contagious}} if for every $\varepsilon>0$, $G$ admits an $\varepsilon$-elaboration whose unique ex-ante equilibrium satisfies $\sigma^*_i( a^*_i|t_i)=1$ for all $i\in I$, $t_i\in T_i$.

Under the ex-ante formulation, ambiguity continues to magnify contagion relative to probabilistic uncertainty. Indeed, this formulation leads to an extreme contagion result, paralleling that under interim optimality and maximum-likelihood updating:

\begin{prop}\label{prop:ex-ante}
Any strict Nash equilibrium $a^*$ of $G$ is ex-ante contagious. 
\end{prop}

Thus, there is a sense in which players' incentives can be more sensitive to small payoff ambiguity under ex-ante than under interim optimality. One novel channel under ex-ante optimality is that players' worst-case beliefs, and hence incentives, are sensitive to the cardinal specification of utilities in an elaboration. For example, if utility functions take the form $v_i(\cdot, t_i, t_{-i})=u_i(\cdot)+ b_i(t_i)$ for some $b_i:T_i\to\mathbb R$, then $b_i (\cdot)$ affects $i$'s incentives under ex-ante optimization, but not under interim optimization.

\section{Concluding Remarks}

We conclude by discussing some directions for future research. First, whereas the unique-implementation application in Section~\ref{sec:contracts} introduced higher-order ambiguity into complete-information games via ambiguous payment schemes, it is also worthwhile to consider information design settings where players share a common prior about a payoff-relevant fundamental $\theta$ and a designer can introduce higher-order ambiguity via ambiguous signals about $\theta$. Several papers study the design of probabilistic signal structures that uniquely implement a target action profile via a process of contagion \citep[e.g.,][]{moriya2020, hoshino2022, morris2024}.\footnote{Contagion is also used in information design under worst-case equilibrium selection \citep[e.g.,][]{li2023, inostroza2025}.}  Our analysis suggests that allowing for ambiguous signals may further facilitate unique implementation by expanding the class of target action profiles for which contagion can be induced.

Second, to illustrate the contrast with probabilistic uncertainty, we identified simple and economically interpretable sufficient conditions for contagion and robustness under ambiguous elaborations. However, it would be valuable to obtain sharper characterizations. As noted in Section~\ref{sec:strategic}, even under probabilistic elaborations, if-and-only-if conditions for contagion and robustness are not known for general games. A promising direction could be to analyze structured classes of games, such as binary-action supermodular games, for which \cite{oyama2020} obtain sharp characterizations under probabilistic uncertainty. Under ambiguity, such characterizations would additionally depend on the class of updating rules.


\newpage

\appendix

\noindent{\Large {\bf Appendix}}

\vspace{-3mm}
\section{Preliminaries}\label{app:gen-belief}
\subsection{Generalized Operators}
We extend the ambiguous $p$-belief and $p$-justifiability operators from Sections~\ref{sec:CB} and \ref{sec:ML} to more general ${\bf f}$-operators \citep{morris2007, oyama2020}. For each agent $i$, consider a function $f_i : 2^{I \setminus \{i\} } \to \mathbb{R}$ such that 
\begin{equation}\label{eq:f-ass}
 f_i (\emptyset) < 0 \quad \text{ and } \quad f_i (T) \leq f_i (S) \text{ for all } T \subseteq S \subseteq I \setminus \{i\} .
\end{equation}
Let ${\bm f} = (f_i)_{i \in I}$. Throughout the appendix, we treat product events $E \subseteq T$ as profiles $E = (E_i)_{i \in I}$ of their projections onto each coordinate $i$. As such, the inclusion $(E_i)_{i\in I}\subseteq (E'_i)_{i\in I}$ means that $E_i\subseteq E_i'$ for each $i \in I$, the intersection $(E_i)_{i\in I}\cap (E'_i)_{i\in I}$ equals $(E_i\cap E_i')_{i\in I}$, and $(E_i)_{i \in I} = \emptyset$ if and only if $E_i = \emptyset$ for all $i \in I$.

Given any type space $(T, \cP)$, any belief-updating rule satisfying Assumption~\ref{asp:ML}, and any product event $E =  (E_j)_{j\in I}$, define the \textbf{\textit{individual ${\bf f}$-belief}} and \textbf{\textit{common ${\bf f}$-belief}} events by 
\begin{align*}
 B^{f_i}_i(E) &=\left\{ t_i\in E_i: \sum_{t_{-i} \in T_{-i}} P_i(t_{-i}) f_i \left( \left\{ j \neq i : t_j \in E_j \right\} \right) \geq 0 \;  \forall P_i\in \cP_{t_i}\right\},\\
 B^{\bf f}(E) &= (B^{f_i}_i(E))_{i\in I}, \;\; (B^{\bf f})^0 (E) = E, \;\;  (B^{\bf f})^k (E) =( B^{f_i}_i (( B^{\bf f})^{k-1} (E)))_{i \in I} \, \forall k\geq 1, \\
 C^{\bf f}_i(E) &= \bigcap_{k\in\mathbb N} (B^{\bf f})_i^k(E), \;\; C^{\bf f}(E) = (C_i^{\bf f}(E))_{i\in I}.
\end{align*}

Define the \textbf{\textit{individual ${\bf f}$-justifiability}} and \textbf{\textit{common ${\bf f}$-justifiability}} events by
\begin{align*}
\overline{B}^{f_i}_i(E) &=\left\{ t_i\in E_i: \sum_{t_{-i} \in T_{-i}} P_i(t_{-i}) f_i \left( \left\{ j \neq i : t_j \in E_j \right\} \right) \geq 0 \;  \text{ for some } P_i\in \cP_{t_i}\right\}, \\
\overline B^{\bf f}(E) &= (\overline B^{f_i}_i(E))_{i\in I}, \;  (\overline B^{\bf f})^0 (E) = E, \;\;  (\overline B^{\bf f})^k (E) =( \overline B^{f_i}_i (( \overline B^{\bf f})^{k-1} (E)))_{i \in I} \, \forall k\geq 1, \\
  \overline C^{\bf f}_i(E) &= \bigcap_{k\in\mathbb N} (\overline B^{\bf f})_i^k(E), \;\; \overline C^{\bf f}(E) = (\overline C_i^{\bf f}(E))_{i\in I}.
\end{align*}
If each $f_i$ is defined by $f_i (S) =  \mathbf{1}_{\{S = I \setminus \{i\} \}} -p$ for all $S \subseteq I\setminus \{i\}$, these definitions reduce to the individual and common $p$-belief and $p$-justifiability definitions in Sections~\ref{sec:CB} and \ref{sec:ML}.

The following lemma summarizes some basic properties of ${\bf f}$-belief and ${\bf f}$-justifiability that we will use throughout the appendix without explicit reference. We omit the proof, which follows standard arguments.
Call a product event $E$ \textit{\textbf{lower $\bf f$-evident}} (resp., \textit{\textbf{upper $\bf f$-evident}}) if $E_i \subseteq B_i^{f_i}(E)$ (resp., $E_i \subseteq \overline{B}_i^{f_i}(E)$) for all $i \in I$.

\begin{lem}\label{lem:basic} 
\begin{enumerate}
\item If $E\subseteq E'$, then $B^{f_i}_i(E)\subseteq B^{f_i}_i(E')$ and $\overline B^{f_i}_i(E)\subseteq \overline B^{f_i}_i(E')$.
\item If $(E^n)_{n\in{\mathbb N}}$ is a sequence of events that is weakly decreasing with respect to set inclusion, then $B^{f_i}_i(\bigcap_{n\in{\mathbb N}}E^n)=\bigcap_{n\in{\mathbb N}}B^{f_i}_i(E^n)$ and $\overline B^{f_i}_i(\bigcap_{n\in{\mathbb N}}E^n)=\bigcap_{n\in{\mathbb N}}\overline B^{f_i}_i(E^n)$. 

\item $C^{\mathbf{f}}(E)$ (resp.\ $\overline{C}^{\mathbf{f}} (E)$) is
the largest lower (resp.\ upper) $\bf f$-evident event contained in $E$.
\end{enumerate}
\end{lem}

\subsection{Strategic Role of Common Belief and Justifiability}

The following lemma extends to settings with ambiguity the well-known result that a $p$-dominant equilibrium of a complete-information game $G$ can be played in an equilibrium of an elaboration of $G$ whenever there is common $p$-belief in normal types. It also establishes an analogous relationship between $p$-secure equilibria (as defined in Section~\ref{sec:robustness}) and the weaker notion of common $p$-justifiability.

\begin{lem}\label{lem:MS} \
Fix any $p\in (0,1)$. 
\begin{enumerate}
\item Suppose $\hat a$ is a $p$-dominant equilibrium of $G$.  Then for any $\varepsilon$-elaboration of $G$, there exists an equilibrium $\sigma$ with $\sigma(\hat a|t)=1$ at all $t\in C^p(T^*)$. 

\item Suppose $\hat a$ is a $p$-secure equilibrium of $G$. Then for any $\varepsilon$-elaboration of $G$, there exists an equilibrium $\sigma$ with $\sigma(\hat a|t)=1$ for all $t\in \overline C^p(T^*)$. 
\end{enumerate}

\end{lem}

\begin{proof}
For the first part, let $\hat\Sigma_i$ be the set of $i$'s strategies such that $\sigma_i(\hat a_i|t_i)=1$ for all $t_i\in  C^{p}_i(T^*)$. Since $C^p(T^*)$ is $p$-evident (in the sense of lower $\mathbf{f}$-evidence defined above), $i$'s best response against any strategy profile in $\hat\Sigma_{-i}$ is in $\hat\Sigma_i$. Thus, by applying Kakutani's fixed-point theorem to the best-response correspondences restricted to strategy profiles with $\sigma_i \in \hat \Sigma_i$ for all $i$, we have an equilibrium $\hat\sigma$ such that $\hat\sigma(\hat a|t)=1$ for all $t\in C^{p}(T^*)$. The proof of the second part follows an analogous argument, by considering the best-response correspondences restricted to strategy profiles where $\sigma(\hat a|t)=1$ for all $t\in \overline C^p(T^*).$
\end{proof}

\section{Proofs for Section~\ref{sec:main}}

\subsection{Proof of Theorem~\ref{thm:CP-amb}}\label{app:CP-amb}

We prove the following generalization of Theorem~\ref{thm:CP-amb} to ${\bf f}$-operators. 

\begin{thm}\label{thm:CP-general}
Fix any $\varepsilon \in (0, 1)$. There is a type space $(T, \cP)$ and product event $E$ such that $P(E)=P'(E) \geq 1-\varepsilon$ for all $P, P'\in \cP$ but $C^{\bf f}_i(E)=\emptyset$ for all functions ${\bf f} = (f_i)_{i \in I}$ satisfying (\ref{eq:f-ass}), all $i\in I$, and all updating rules satisfying Assumption~\ref{asp:ML}.
\end{thm}

\begin{proof}
Pick $N\in\mathbb N$ large enough that $\frac{2}{N+1}\leq\varepsilon$. Let $T_i=\{0,1,\ldots, N\}$ for each $i\in I$. 
Let $\cP=\{P^i : i\in I\}$, where each $P^i$ is defined by $P^i(t)=\frac{1}{N+1}$ for any $t\in T$ such that $t_j\equiv t_i -1$ for all $j\not=i$, where $\equiv$ denotes equality modulo $N+1$.\footnote{Thus, $P^i$ assigns equal probability to the $N+1$ type profiles $(t_i = N, t_{-i}  = \mathbf{N-1})$, $(t_i = N-1, t_{-i}  = \mathbf{N-2})$, $\ldots$, $(t_i = 1, t_{-i}  = \mathbf{0})$, $(t_i = 0, t_{-i}  = \mathbf{N})$ and probability zero to all other $t$.}
Observe that ${\rm marg}_{T_i}P^j$ is the same (viz., uniform on $\{0, 1, \ldots, N\}$) for all $j$. Thus, by Assumption~\ref{asp:ML}, $\cP_{t_i} = \cP^{\rm FB}_{t_i}$ for every $t_i$. 

Consider the product event $E = (E_i)_{i \in I}$ with $E_i = \{1,2,\ldots, N\}$ for each $i$. Then for all $i \in I$, $P^i (E)=\frac{N-1}{N+1} \geq 1 - \varepsilon$. At the same time, for each $i \in I$ and $n \in\{0,1,\ldots, N-1\}$, we have 
 \begin{equation}\label{eq:conditional-ML'}
P^i \left (\{t_{j}=n \forall j\not=i\} \mid t_i = n+1 \right)=1.
\end{equation}
Take any functions ${\bf f} = (f_i)_{i \in I}$ satisfying (\ref{eq:f-ass}). To show that $C^{\bf f}_i (E) = \emptyset$ for each $i\in I$, it suffices to prove inductively that for all $n \in \mathbb{N}$, 
\[
\left(B^{\bf f} \right)^n_i (E) \subseteq \left\{ n+1,\ldots, N\right\} \text{ for each } i,
\] 
where we set $\left\{ n+1,\ldots, N\right\} = \emptyset$ for all $n \geq N$. The claim is immediate for $n = 0$, as $\left( B^{\bf f} \right)^0 (E) = E $. Suppose the claim holds for some $n$ with $0 \leq n \leq N-1$. Consider any $i$. Then for $t_i=n+1$, (\ref{eq:conditional-ML'}) and the inductive hypothesis imply that
\[
\sum_{t_{-i} \in T_{-i}} P^i \left( t_{-i}   \mid t_i  \right) f_i \left( \left\{ j \neq i : t_j \in \left(B^{\bf f} \right)^n_j  (E) \right\} \right)= f_i(\emptyset)<0.
\]
Hence, $B^{ f_i}_i   \left( \left(B^{\bf f} \right)^n (E) \right) \subseteq\left\{ n+2,\ldots, N\right\}$. Thus, $\left(B^{\bf f} \right)^{n+1}_i (E) \subseteq\left\{ n+2,\ldots, N\right\}$. This establishes the claim for any $n \leq N$ and, in particular, shows that $\left(B^{\bf f} \right)^N_i (E) = \emptyset$ for all $i$. Then we also have $\left(B^{\bf f} \right)^n_i (E) = \emptyset$ for all $i$ and $n > N$.
\end{proof}

\subsection{Proof of Theorem~\ref{prop:contagious}}\label{app:contagious}
Since $a^*$ is secure, there is $\gamma < 1$ such that for all $i$, $a_i \neq a^*_i$, and $Q_i \in \Delta(A_{-i})$ with $Q_i(a^*_{-i}) \geq \gamma$, we have 
\begin{equation}\label{eq:gamma2}
\min_{a_{-i} \in A_{-i}} u_i (a^*_i, a_{-i})>\sum_{a_{-i} \in A_{-i}} Q_i(a_{-i}) u_i (a_i, a_{-i}).
\end{equation}
For each $i$, define $f_i$ by 
$$f_i (S) = \begin{cases} \gamma - 1 & \text{ if } S = \emptyset\\
\gamma & \text{ if } S \neq \emptyset.
\end{cases}
$$

Consider any $\varepsilon > 0$. By Theorem~\ref{thm:CP-general}, there is a type space $(T, \cP)$ and a product event $E$ such that $P(E)\geq 1-\varepsilon$ for all $P\in \cP$ but $C^{\bf f}_i(E)=\emptyset$ for each $i\in I$.   
Define a bounded payoff function $v_i : A \times T \to \mathbb R$ for each $i$ such that $v_i (\cdot , (t_i, t_{-i}) )=  u_i (\cdot) $ for all $t_i \in E_i$ and $t_{-i} \in T_{-i}$ and such that $a^*_i$ is  strictly dominant  for all $t_i \notin E_i$. Then $(T, {\cP}, v)$ is an $\varepsilon$-elaboration of $u$.

Consider any equilibrium $\sigma$ of this elaboration. 
Let $\hat T_i=\{t_i \in T_i: \sigma_i(a^*_i|t_i)<1\}$ for each $i\in I$, and let $\hat T=(\hat T_i)_{i\in I}$. 
Note that $\sigma_i(a^*_i|t_i)=1$ for each $t_i\not\in E_i$, and thus $\hat T\subseteq E$.
Below we show that $\hat T_i=\emptyset$ for each $i\in I$, so that $\sigma (a^* \mid t ) = 1$ for all $t \in T$.

For each $i\in I$ and $t_i\in \hat T_i$, there is $a_i\not=a^*_i$ with $\min_{P_i\in\cP_{t_i}}\sum_{t_{-i}}P_i(t_{-i})u_i(a_i, \sigma_{-i}(t_{-i}))\geq \min_{P_i\in\cP_{t_i}}\sum_{t_{-i}}P_i(t_{-i})u_i(a^*_i, \sigma_{-i}(t_{-i}))$. By (\ref{eq:gamma2}), this implies $P_i((T_j\setminus \hat T_j)_{j\not=i})\leq \gamma$ for all $P_i\in\cP_{t_i}$, and hence $t_i\in B^{f_i}_i(\hat T)$.  Thus, $\hat T$ is lower ${\bf f}$-evident. Since $\hat T\subseteq E$, we have $\hat T\subseteq C^{\bf f}(E)$.  Then $\hat T_i=\emptyset$ for each $i \in I$, as $C^{\bf f}_i(E)=\emptyset$ by construction. \qed

\subsection{Proof of Proposition~\ref{prop:CP-ML}}

\subsubsection{First part}

Take any $p<\frac{1}{|I|}$,  type space $(T, \cP)$, product event $E$, and $P\in \cP$. Let $\hat B^{p}(\cdot)$ and $\hat C^{p}(\cdot)$ be the ${p}$-belief operator and common ${p}$-belief operator under the unambiguous type space $(T, \{P\})$. 
Then  $\hat B^{p}(\cdot)\subseteq \overline B^{p}(\cdot)$.  By monotonicity and continuity of  $\hat B^{p}(\cdot)$ and $\overline B^{ p}(\cdot)$ (see Lemma~\ref{lem:basic}), standard arguments imply $\hat C^{p}(\cdot)\subseteq \overline C^{ p}(\cdot)$. 
Thus, by \cites{kajii1997} critical path theorem,  we have $P (\overline C^{ p}(E)) \geq P(\hat C^{p}(E)) \geq 1-(1-P(E))\frac{1-p}{1-|I|p}$, as claimed.

\subsubsection{Second part}
We prove the following more general version based on $\bf f$-operators. 

\begin{prop}\label{prop:CP-ML'}
Consider maximum-likelihood updating. For all $\varepsilon \in (0,1)$, there is a type space $(T, \cP)$ and a product event $E$ such that $P(E)=P'(E) \geq 1-\varepsilon$ for all $P, P'\in \cP$ but $\overline C^{\bf f}_i(E)=\emptyset$ for all functions ${\bf f} = (f_i)_{i \in I}$ satisfying (\ref{eq:f-ass}) and all $i\in I$. 
\end{prop}

\begin{proof}
Fix any $\varepsilon \in (0, 1)$. Let $T = (T_i)_{i\in I}$, with $T_i  = \mathbb{N}$ for all $i$. 
Choose $r\in (1-\varepsilon, 1)$, and set 
\[
c=\frac{(1-\varepsilon)(1-r)}{r}, \;\; \quad d=\varepsilon-c>0.
\]
Let $\cP=\{P^i : i\in I\}$, where each $P^i$ is defined by
\[
P^i(t)=
\begin{cases}
cr^{n} & \text{ if } \exists  n\in{\mathbb N} \text{ such that } t={\mathbf n}+{\mathbf 1}_i \\
d & \text{ if } t={\mathbf 0}\\
0 &\text{ otherwise}, 
\end{cases}
\]
where ${\mathbf 1}_i \in \mathbb{R}^I$ denotes the $i$th standard basis vector. Consider the product event $E = \{1,2,\ldots \}^I$. Then for all $i \in I$,
\[
P^i (E)=P^i \left( \{\mathbf{n}+{\bf 1}_{i} : n > 0\} \right) =1 -\varepsilon.
\]
At the same time, for each $n \in \mathbb{N}$ and $t_i = n+1$, we have ${\cal P}^{\rm ML}_{t_i}=\{P^i(\cdot \mid t_i)\}$. This is because for all $j \neq i$,
\[
P^i(t_i) = cr^n >cr^{n+1}=P^j(t_i).
\] 
  Moreover, $P^i (\cdot \mid t_i)$ satisfies (\ref{eq:conditional-ML'}).

Take any functions ${\bf f} = (f_i)_{i \in I}$ satisfying (\ref{eq:f-ass}). To show that $\overline C^{\bf f}_i (E) = \emptyset$ for each $i\in I$, it suffices to prove inductively that for all $n \in \mathbb{N}$, $\left(\overline B^{\bf f} \right)^n_i (E) \subseteq \{n+1,n+2,\ldots \}$ for each $i$. The claim is immediate for $n = 0$, as $\left(\overline B^{\bf f} \right)^0 (E) = E$. Suppose the claim holds for some $n \geq 0$. Consider any $i$. Then for $t_i=n+1$, (\ref{eq:conditional-ML'}) and the inductive hypothesis imply that
\[
\sum_{t_{-i} \in T_{-i}} P^i \left( t_{-i}   \mid t_i  \right) f_i \left( \left\{ j \neq i : t_j \in \left(\overline B^{\bf f} \right)^n_j  (E) \right\} \right)= f_i(\emptyset)<0.
\]
Hence, $\overline B^{ f_i}_i   \left( \left(\overline B^{\bf f} \right)^n (E) \right) \subseteq \{n + 2, n+3,\ldots \}$, as required.  
\end{proof}

\subsection{Proof of Proposition~\ref{prop:contagious-ML}}

Consider any strict Nash equilibrium $a^*$ of $G$. Then we can pick $\gamma < 1$ such that for all $i$, $a_i \neq a^*_i$, and $Q_i \in \Delta(A_{-i})$ with $Q_i(a^*_{-i}) \geq \gamma$, we have 
\begin{equation}\label{eq:gamma}
\sum_{a_{-i} \in A_{-i}} Q_i(a_{-i}) u_i (a^*_i, a_{-i})>\sum_{a_{-i} \in A_{-i}} Q_i(a_{-i}) u_i (a_i, a_{-i}).
\end{equation}

Define $\bf f$ as in the proof of Theorem~\ref{prop:contagious}. Construct a type space $(T, {\cal P})$ and product event $E$ as in the proof of Proposition~\ref{prop:CP-ML'} and use this to construct an $\varepsilon$-elaboration $(T, {\cal P}, v)$ as in the proof of Theorem~\ref{prop:contagious}. Let $\hat T_i=\{t_i: \sigma_i(a^*_i|t_i)<1\}$ for each $i\in I$, and let $\hat T= (\hat T_i)_{i\in I}$. For each $i\in I$ and $t_i\in \hat T_i$, there exist $a_i\not=a^*_i$ and $P_i\in {\cal P}_{t_i}$ with $\sum_{t_{-i}}P_i(t_{-i})u_i(a_i, \sigma_{-i}(t_{-i}))\geq \sum_{t_{-i}}P_i(t_{-i})u_i(a^*_i, \sigma_{-i}(t_{-i}))$. By (\ref{eq:gamma}), this implies $P_i((T_j\setminus \hat T_j)_{j\not=i} )\leq \gamma$, and hence $t_i\in \overline B^{f_i}_i(\hat T)$.  Thus, $\hat T$ is upper ${\bf f}$-evident. Since $\hat T\subseteq E$, we have $\hat T\subseteq \overline C^{\bf f}(E)$.  Then $\hat T_i=\emptyset$ for each $i \in I$, as $\overline{C}^{\bf f}_i(E)=\emptyset$ by construction. Hence, $\sigma (a^* \mid t ) = 1$ for all $t \in T$. \qed

\section{Proofs for Section~\ref{sec:implication}}

\subsection{Proof of Corollary~\ref{cor:FB/ML}}

The first part follows from the second part of Lemma~\ref{lem:MS} and the first part of Proposition~\ref{prop:CP-ML}.
The second part follows from  Proposition~\ref{prop:contagious-ML}. \qed

\subsection{Partial Implementation}

The following lemma shows that to partially implement $a^*$, the principal can without loss of optimality focus on deterministic payment schemes. 

\begin{lem}\label{lem:partial} For every $a^*\in A$, $W^{a^*}_{{\rm partial}}= \sum_{i \in I} w^*_{i, {\rm partial}}$, where $(w^*_{i, {\rm partial}})$ is the deterministic payment scheme given by $w^*_{i, {\rm partial}}=\max_{a_i\in A_i}u_i(a_i, a^*_{-i})-u_i(a^*)$ for all $i$.
\end{lem}

\begin{proof}
To see that $W^{a^*}_{{\rm partial}}\leq \sum_{i\in I} w^*_{i, {\rm partial}}$, note that the deterministic payment scheme where each agent $i$ receives $w^*_{i, {\rm partial}} $ for playing $a^*_i$ induces $a^*$ as an equilibrium. 

To show that $W^{a^*}_{{\rm partial}}\geq \sum_{i\in I} w^*_{i, {\rm partial}}$, take any payment scheme $(T, {\cal P}, w)$ that partially implements $a^*$. Then for each type $t_i$, we must have $u_i(a^*)+w_i(t_i)\geq \max_{a_i\not=a^*_i}u_i(a_i, a^*_{-i})$, i.e., $w_i(t_i)\geq \max_{a_i\not=a^*_i}u_i(a_i, a^*_{-i})-u_i(a^*)$. Since $w_i(\cdot)$ is non-negative, we also have $w_i(t_i)\geq \max_{a_i\in A_i}u_i(a_i, a^*_{-i})-u_i(a^*) = w^*_{i, {\rm partial}}$.  Thus, for every $P\in \cP$, we have $\sum_{t\in T} P(t)\sum_{i\in I}w_i(t_i)\geq \sum_{i\in I} w^*_{i, {\rm partial}}$. \end{proof}

\subsection{Proof of Corollary~\ref{lem:FB bound}}\label{app:FB-contract}

Fix any target action profile $a^*\in A$ satisfying condition (\ref{eq:general-impl}). Since $W^{a^*}_{{\rm partial}}\leq W^{a^*}$ by definition, it suffices to show that for any $\kappa>0$, there is ${\cal C}$ satisfying (\ref{eq:full}) and (\ref{eq:consistency}) such that $\min_{P\in \cP} \sum_{t\in T} P(t)\sum_{i\in I}w_i(t_i)\leq W^{a^*}_{{\rm partial}} +\kappa$.

Let $w^*_i = \max_{a_i\in A_i}u_i(a_i, a^*_{-i})-u_i(a^*)+\frac{\kappa}{2|I|}$. Then $a^*$ is a strict Nash equilibrium of the game with modified utility functions $\tilde u_i(a)=u_i(a)+w^*_i{\bf 1}_{a_i=a^*_i}$, and in this game $a^*$ is secure by (\ref{eq:general-impl}). 
Choose $K>0$  large enough that $u_i(a^*_i, a_{-i})+K>u_i(a_i, a_{-i})$ for all $i\in I$, $a_i\not=a_i^*$, and $a_{-i}\in A_{-i}$.

For  any $\varepsilon\leq \frac{\kappa}{2|I|K}$, let $(T, {\cP})$ and $E$ be the type space and event  used in the proof of Theorem~\ref{thm:CP-general}. The construction in that proof implies that $\sum_{i\in I}P(((T_i\setminus E_i), T_{-i}))=\sum_{i\in I}P'(((T_i\setminus E_i), T_{-i})) \leq |I| \varepsilon$ for all $P,P'\in\cP$. Let ${\cal C}=(T, {\cP}, w)$, where for each $i\in I$, 
\begin{equation}
w_i(t_i)=
\begin{cases}\label{eq:w construction}
w^*_i  \text{ for } t_i\in E_i  \\
w^*_i+K \text{ for } t_i\not\in E_i.
 \end{cases}
\end{equation}
 Note that the incomplete-information game induced by $\cC$ is an $\varepsilon$-elaboration of $(A, \tilde u)$ and that by the choice of $K$, $a^*_i$ is strictly dominant for each $t_i\not\in E_i$. Moreover, the proof of 
Theorem~\ref{prop:contagious} shows that $\cC$ has $a^*$ as its unique equilibrium.

Since $\sum_{i\in I}w^*_i=W^{a^*}_{\rm partial}+\kappa/2$, the total expected payment under each $P \in {\cP}$  is 
\[
\sum_{t\in T} P(t)\sum_{i\in I}w_i(t_i) =\sum_{i\in I} \left(w^*_i+ P(((T_i\setminus E_i), T_{-i})) K\right) \leq W^{a^*}_{\rm partial}+\frac{\kappa}{2}+ |I|\varepsilon K \leq W^{a^*}_{\rm partial}+\kappa.
\]
Moreover, by construction of $(T, {\cP}, E)$, the total expected payment is the same for all $P \in {\cal P}$, so (\ref{eq:consistency}) holds. This proves that for any $\kappa > 0$, $W^{a^*}  \leq W^{a^*}_{\rm partial}+\kappa$. \qed

\subsection{Unambiguous Implementation}\label{app:unambig-impl}

To determine the unique implementation cost under unambiguous payment schemes in Example~\ref{ex:contracts}, we show a more general result for binary-action supermodular games: Here, for each player $i \in I$, we have $A_i=\{0,1\}$, and $i$'s payoff difference between action $1$ vs.\ action $0$, i.e., $d_i(S) =  u_i \left(1,\mathbf 1_S, \mathbf{0}_{I \setminus (S \cup \{i\})} \right) -  u_i\left(0,\mathbf 1_S, \mathbf{0}_{I \setminus (S \cup \{i\})}\right)$, is nondecreasing with respect to set inclusion in the set $S$ of opponents playing action 1.

The following result follows from arguments analogous to those in \cite{morris2022}, so we relegate the proof to Online Appendix~\ref{app:unambig-impl-pf}. Let $\Pi$ be the set of permutations of $I$. For $\gamma\in\Pi$, let
$S_{-i}(\gamma)$ be the set of players that precede $i$ in $\gamma$.

\begin{prop}\label{prop:prob}
In any binary-action supermodular game, we have
\begin{equation}
W_{\mathrm{prob}}^{\mathbf 1}
=
\min_{\substack{z\in\mathbb R_+^I\\ \rho\in\Delta(\Pi)}}
\left\{
 \sum_{i\in I}z_i:
 z_i+
 \sum_{\gamma\in\Pi}
 \rho(\gamma)d_i\bigl(S_{-i}(\gamma)\bigr)
 \ge 0
 \quad\text{for every }i\in I
\right\}.
\label{eq:prob-cost-sequential-obedience}
\end{equation}
\end{prop}

The inequality condition in (\ref{eq:prob-cost-sequential-obedience}) requires that, when each player $i$ receives the deterministic payment $z_i$ for playing action 1, action profile ${\bf 1}$ satisfies a sequential obedience condition; in symmetric two-player games (as in Example~\ref{ex:contracts}), this condition is equivalent to ${\bf 1}$ being weakly risk-dominant \citep[see][]{oyama2020}. 

Applied to the game in Example~\ref{ex:contracts} (and swapping the labels for actions $0$ and $1$), Proposition~\ref{prop:prob} implies that $W^{\bf 0}_{\rm prob} = \max\{b(1)-b(0)-2c, 0\}$, as making ${\bf 0}$ weakly risk-dominant requires $z_i \geq 0.5(b(1)-b(0))-c$ for each $i$.
Likewise, to show that under unambiguous insurance contracts $\overline W^{\bf 1}_{\rm prob}=  \max \{2c-2b(1)+b(0), 0\}$, one can apply Proposition~\ref{prop:prob} to the modified game where action $1$ is safe with payoff $b(1) -c$.

\section{Proofs for Section~\ref{sec:discussion}}

\subsection{Proof of Proposition~\ref{prop:ex-ante}}

Fix any strict Nash equilibrium $a^*$ of $G$, and any $\varepsilon > 0$. We construct an $\varepsilon$-elaboration based on the type space used in the second part of Proposition~\ref{prop:CP-ML}, where $T_i =  \mathbb{N}$ for all $i$ and $\cP=\{P^i : i\in I\}$ has the property that $P^i(t_i=0)<P^j(t_i=0)$ for all $i \neq j$.

Let $T^*_i =  \{t_i \in T_i : t_i \geq 1\}$ be the set of $i$'s normal types. 
We specify the payoff of the non-normal type $t_i=0$ as
\[
v_i(a, (0, t_{-i}))={\mathbf 1}_{a_i=a^*_i}+M_i,
\]
where $M_i>0$ is large enough that for all $j\not=i$,
\[
P^i(t_i=0)(M_i+1)
+\bigl(1-P^i(t_i=0)\bigr)\max_{a\in A}u_i(a)
<
P^j(t_i=0)M_i
+\bigl(1-P^j(t_i=0)\bigr)\min_{a\in A}u_i(a).
\]
This inequality ensures that player $i$'s ex-ante worst-case payoff is always evaluated based on the prior $P^i$, i.e., for any strategy profile $\sigma$, 
\[
\min_{P\in \cP}\sum_{t\in T}P(t)v_i \left(\sigma(t), t \right)=\sum_{t\in T}P^i(t)v_i \left(\sigma(t), t \right).
\]
 Thus, against any opponent strategies $\sigma_{-i}$, $i$'s ex-ante best response 
\[
\sigma^*_i\in\argmax_{\sigma_i}\min_{P\in \cP}\sum_{t\in T}P(t)v_i \left(\sigma_i(t_i), \sigma_{-i}(t_{-i}), t_i, t_{-i} \right)
\]
satisfies interim optimality for each $t_i \in T_i$ with respect to the singleton posterior belief $\{P^i( \cdot \mid t_i) \}$, i.e., 
\[
\sigma^*_i (t_i)\in \argmax_{\alpha_i\in \Delta(A_i)}\sum_{a_i}\alpha_i(a_i)\sum_{t_{-i}}P^i(t_{-i}|t_i)v_i \left(a_i, \sigma_{-i}(t_{-i}), t_i, t_{-i} \right). 
\]
By construction, $a_i^*$ is a strictly dominant action for $t_i=0$. Then an analogous argument as in Proposition~\ref{prop:contagious-ML} implies that the unique ex-ante equilibrium $\sigma^*$ satisfies $\sigma^* (a^* \mid t) = 1$ for all $t \in T$. \qed

\section{Omitted Details}

\subsection{Role of Conservativeness in Belief Updating}\label{app:CPJ-FB}

To shed more light on the role of agents' belief updating, we extend the comparison of full Bayesian and maximum-likelihood updating in Section~\ref{sec:ML} to an intermediate class of updating rules.  Under \textit{\textbf{$\beta$-mixture updating}}, for all type spaces $(T, \cP)$, $i \in I$, and $t_i \in T_i$, we have
\begin{equation}\label{eq:hybrid}
\cP_{t_i}=\beta {\rm co} \, {\cP}_{t_i}^{\rm FB}+(1-\beta){\rm co} \, {\cP}_{t_i}^{\rm ML}.
\end{equation}
Thus, $t_i$'s set of posteriors contracts the closed convex hull of the full-Bayesian posterior set towards that of the maximum-likelihood posterior set. Parameter $\beta$ captures the extent of conservativeness, in the sense that the sets of posteriors $\cP_{t_i}$ are increasing in $\beta$ with respect to set inclusion. Convexification does not affect maxmin payoffs or the belief and justifiability operators, so $\beta=1$ and $\beta=0$ are
behaviorally equivalent to full Bayesian and maximum-likelihood updating, respectively.\footnote{While this updating rule can violate Assumption~\ref{asp:ML}, this violation does not play a role in the analysis. \cite{epstein2007}  and \cite{cheng2019} propose similar hybrids of full Bayesian and maximum-likelihood updating, for which similar results can be obtained. }

 The following proposition generalizes Proposition~\ref{prop:CP-ML}. The first part provides a weaker version of the critical path theorem with respect to common justifiability, which becomes stronger when $\beta$ is higher, i.e., updating is more conservative.   The second part exhibits failures of the critical path theorem, which are  more extreme under smaller $\beta$, i.e., less conservative updating rules:
 
\begin{prop}\label{prop:beta} Under $\beta$-mixture updating, the following holds:  
\begin{enumerate}
\item For every type space $(T, \cP)$, product event $E$, $P \in \cP$, and $p<\frac{1}{|I|}$, we have
\[
P(\overline C^{\beta p}(E)) \geq 1-(1-P(E))\frac{1-p}{1-|I|p}.
\] 

\item For any $\varepsilon \in (0, 1)$, there is a type space $(T, \cP)$ and product event $E$ such that $P(E)=P'(E) \geq 1-\varepsilon$ for all $P, P' \in\cP$ but $\overline C^{p}(E)=\emptyset$ for all $p>\beta$. 
\end{enumerate}
\end{prop}
 
 \begin{proof} For the first part, let $\hat B^{p}(\cdot)$ and $\hat C^{p}(\cdot)$ be the ${p}$-belief operator and common ${p}$-belief operator under the unambiguous type space $(T, \{P\})$.  For each $t_i\in T_i$ and product event $D\subseteq T$ such that $P(D_{-i}|t_i)\geq p$, we have $Q_i(D_{-i})\geq\beta p$ for some $Q_i\in \cP_{t_i}$. Thus, $\hat B^{p}(\cdot)\subseteq \overline B^{\beta p}(\cdot)$.  Then the desired conclusion follows by the same argument as in the proof of the first part of Proposition~\ref{prop:CP-ML}.
The proof of the second part is analogous to the second part of Proposition~\ref{prop:CP-ML}. \end{proof}

Proposition~\ref{prop:beta} implies the following corollary:

  \begin{cor}\label{cor:beta}
Under $\beta$-mixture updating, for any complete-information game $G$:

\begin{enumerate}
\item If $a^* \in A$ satisfies $\beta\min_{a_{-i}}u_i(a^*_i, a_{-i})+(1-\beta)u_i(a^*)>\max_{a_i\not=a^*_i}u_i(a_i, a^*_{-i})$ for every $i\in I$, then $a^*$ is contagious.

\item If $\hat a$ is a $p$-secure equilibrium of $G$ with $p<\beta/|I|$, then $\hat a$ is robust. 
\end{enumerate}

 \end{cor}
 
 \begin{proof}
 The argument for the first part is analogous to the second part of Proposition~\ref{prop:contagious-ML}. 
The second part follows from the first part of Proposition~\ref{prop:beta} and Lemma~\ref{lem:MS}. 
 \end{proof}

The condition in the first part requires that the payoff to the equilibrium action $a^*_i$ evaluated at the worst-case belief in $\beta\Delta(A_{-i})+(1-\beta)\{\delta_{a^*_{-i}}\}$ is better than the payoff to any deviation $a_{i}$ evaluated at the equilibrium belief $\delta_{a^*_{-i}}$. When $\beta = 1$, this reduces to $a^*$ being secure as in Theorem~\ref{prop:contagious}; when $\beta = 0$, this reduces to $a^*$ being a strict Nash equilibrium as in Proposition~\ref{prop:contagious-ML}. 

Extending Corollary~\ref{cor:FB/ML}, the second part shows that a $p$-secure $\hat a$ is robust when $p$ is sufficiently small, where the requirement on $p$ is more demanding the smaller is $\beta$.

\subsection{Details for Remark~\ref{rem:small-ambig}}\label{app:payoff-irrelevant}

As discussed in Remark~\ref{rem:small-ambig}, it is important for our results that while $\varepsilon$-elaborations only feature small prior ambiguity about payoffs, they allow for significant ambiguity about payoff-irrelevant events. Elaborating on 
Remark~\ref{rem:small-ambig}, we show that if $(T, \mathcal{P})$ features small prior ambiguity about \emph{all} events in $T$, then it satisfies an approximate version of the critical path theorem, and as a result, $(T, \mathcal{P})$ does not support the contagion in Theorem~\ref{prop:contagious}.

One natural way to limit the amount of ambiguity about all events in $T$ is to bound the likelihood ratios $\frac{P(t)}{P'(t)}$ across different $P,P'\in\cal P$. When ${\rm diam } \, {\cal P}:=\sup_{P, P'\in {\cal P}, t\in T}\frac{P(t)}{P'(t)}$ is small, we obtain the following approximate version of the critical-path theorem:\footnote{\cite{oyama2010} derive an analogous result in the context of non-common probabilistic priors. We use the conventions $0/0:=1$ and $x/0:=+\infty$ for every $x>0$.}

\begin{prop}\label{prop:bound}
Consider any updating rule satisfying Assumption~\ref{asp:ML}. Take any $r \geq 1$ and $ p < \frac{1}{|I|r}$. Then for any type space $(T, \mathcal{P})$ such that ${\rm diam}\, {\cal P} \leq \sqrt{r}$ and any product event $E$,
\[
P(C^{p}(E)) \geq 1 - \frac{1 - r p}{1 - r |I|p} (1 - P(E)), \; \forall P \in \mathcal{P}.
\]
\end{prop}

\begin{proof}
Fix any type space $(T, \mathcal{P})$ such that  ${\rm diam}\, {\cal P} \leq \sqrt{r}$ and any $P^0 \in \mathcal{P}$.
Let $B^{{p}, 0}$ and $C^{{p}, 0}$ denote the individual and common belief operators for the unambiguous type space $(T, \{P^0\})$.
We have $B^{r{p}, 0}(E) \subseteq B^{{p}}(E)$ and hence $C^{r{p}, 0}(E) \subseteq C^{{p}}(E)$ for any product event $E \subseteq T$:
indeed,
if $t_i \in B^{r{p}, 0}_i(E)$, i.e., $t_i \in E_i$ and $P^0(\{t_i\} \times E_{-i}) \geq P^0(t_i) \times rp$,
then $t_i \in E_i$ and $P'(\{t_i\} \times E_{-i}) \geq P'(t_i) \times p$ for all $P' \in \mathcal{P}$, which ensures 
 $t_i \in B^{{p}}_i(E)$ by Assumption~\ref{asp:ML}-(i).

Since $ rp < 1/{|I|}$ by assumption, \cites{kajii1997} critical path theorem applied to $(T, \{P^0\})$ implies that for any product event $E \subseteq T$,
\[ P^0(C^{r{p}, 0}(E)) \geq 1 - \frac{1 - rp}{1 - r|I|p} (1 - P^0(E)),
\]
which proves the claim as $P^0(C^{{p}}(E)) \geq P^0(C^{r{p}, 0}(E))$.
\end{proof}

As a corollary of Proposition~\ref{prop:bound} and Lemma~\ref{lem:MS}, one also recovers a version of the robustness result in \cite{kajii1997} by bounding payoff-irrelevant ambiguity:

\begin{cor}
Consider any updating rule satisfying Assumption~\ref{asp:ML}. Suppose $a^*$ is a $p$-dominant equilibrium of $G$ with $p<\frac{1}{|I|r}$ for some $r\geq 1$. Then $a^*$ is robust against all elaborations with $
{\rm diam}\, {\cal P} \leq \sqrt{r}$.
\end{cor}

\vspace{2mm}

\footnotesize
\begin{spacing}{0.05}
\bibliographystyle{econometrica}
\bibliography{contagious}

@article{aumann1976,
  title={Agreeing to disagree},
  author={Aumann, Robert J.},
  journal={The Annals of Statistics},
  volume={4},
  number={6},
  pages={1236--1239},
  year={1976},
  publisher={Institute of Mathematical Statistics}
}

@unpublished{morris2022,
	author = {Morris, Stephen and Oyama, Daisuke and Takahashi, Satoru},
	note = {working paper},
	title = {On the joint design of information and transfers},
	year = {2022}}

@article{tang2021,
	author = {Tang, Rui and Zhang, Mu},
	journal = {Journal of Economic Theory},
	pages = {105250},
	title = {Maxmin implementation},
	volume = {194},
	year = {2021}}

@article{KMsurvey,
	author = {Kajii, Atsushi and Morris, Stephen},
	journal = {The Japanese Economic Review},
	number = {1},
	pages = {7--34},
	publisher = {Springer},
	title = {Refinements and higher-order beliefs: a survey},
	volume = {71},
	year = {2020}}

@article{KMsurvey2,
	author = {Kajii, Atsushi and Morris, Stephen},
	journal = {The Japanese Economic Review},
	number = {1},
	pages = {35--41},
	publisher = {Springer},
	title = {Notes on ``Refinements and higher-order beliefs''},
	volume = {71},
	year = {2020}}

@article{bergemann2013,
	author = {Bergemann, Dirk and Morris, Stephen},
	journal = {Foundations and Trends{\textregistered} in Microeconomics},
	number = {3},
	pages = {169--230},
	publisher = {Emerald Publishing Limited},
	title = {An Introduction to Robust Mechanism Design},
	volume = {8},
	year = {2013}}

@article{horie2013,
	author = {Horie, Mayumi},
	journal = {Journal of Mathematical Economics},
	number = {6},
	pages = {467--470},
	title = {Reexamination on updating choquet beliefs},
	volume = {49},
	year = {2013}}

@article{eichberger2007,
	author = {Eichberger, J{\"u}rgen and Grant, Simon and Kelsey, David},
	journal = {Journal of Mathematical Economics},
	number = {7-8},
	pages = {888--899},
	title = {Updating choquet beliefs},
	volume = {43},
	year = {2007}}

@article{billot2000,
	author = {Billot, Antoine and Chateauneuf, Alain and Gilboa, Itzhak and Tallon, Jean-Marc},
	journal = {Econometrica},
	number = {3},
	pages = {685--694},
	title = {Sharing beliefs: between agreeing and disagreeing},
	volume = {68},
	year = {2000}}

@article{kajii2009,
	author = {Kajii, Atsushi and Ui, Takashi},
	journal = {Journal of Economic Theory},
	number = {1},
	pages = {337--353},
	title = {Interim efficient allocations under uncertainty},
	volume = {144},
	year = {2009}}

@article{seidenfeld1993,
	author = {Seidenfeld, Teddy and Wasserman, Larry},
	journal = {The Annals of Statistics},
	number = {3},
	pages = {1139--1154},
	title = {Dilation for sets of probabilities},
	volume = {21},
	year = {1993}}

@article{hanany2020,
	author = {Hanany, Eran and Klibanoff, Peter and Mukerji, Sujoy},
	journal = {American Economic Journal: Microeconomics},
	number = {2},
	pages = {135--187},
	title = {Incomplete information games with ambiguity averse players},
	volume = {12},
	year = {2020}}

@incollection{machina2014,
	author = {Machina, Mark J and Siniscalchi, Marciano},
	booktitle = {Handbook of the Economics of Risk and Uncertainty},
	publisher = {Elsevier},
	title = {Ambiguity and ambiguity aversion},
	volume = {1},
	year = {2014}}

@article{inostroza2025,
	author = {Inostroza, Nicolas and Pavan, Alessandro},
	journal = {Theoretical Economics},
	number = {2},
	pages = {763--813},
	title = {Adversarial coordination and public information design},
	volume = {20},
	year = {2025}}

@article{li2023,
	author = {Li, Fei and Song, Yangbo and Zhao, Mofei},
	journal = {Journal of Economic Theory},
	pages = {105575},
	title = {Global manipulation by local obfuscation},
	volume = {207},
	year = {2023}}

@article{hoshino2022,
	author = {Hoshino, Tetsuya},
	journal = {International Economic Review},
	number = {2},
	pages = {755--776},
	title = {Multi-Agent Persuasion: Leveraging Strategic Uncertainty},
	volume = {63},
	year = {2022}}

@article{moriya2020,
	author = {Moriya, Fumitoshi and Yamashita, Takuro},
	journal = {Journal of Economics \& Management Strategy},
	number = {1},
	pages = {173--186},
	title = {Asymmetric-information allocation to avoid coordination failure},
	volume = {29},
	year = {2020}}

@article{halac2025,
	author = {Halac, Marina},
	journal = {Journal of the European Economic Association},
	number = {3},
	pages = {815--844},
	title = {Contracting for coordination},
	volume = {23},
	year = {2025}}

@article{klibanoff2009,
	author = {Klibanoff, Peter and Marinacci, Massimo and Mukerji, Sujoy},
	journal = {Journal of Economic Theory},
	number = {3},
	pages = {930--976},
	title = {Recursive smooth ambiguity preferences},
	volume = {144},
	year = {2009}}

@article{pei2025,
	author = {Pei, Harry and Strulovici, Bruno},
	journal = {Review of Economic Studies},
	number = {1},
	pages = {476--505},
	title = {Robust implementation with costly information},
	volume = {92},
	year = {2025}}

@article{morris2024,
	author = {Morris, Stephen and Oyama, Daisuke and Takahashi, Satoru},
	journal = {Econometrica},
	number = {3},
	pages = {775--813},
	title = {Implementation via Information Design in Binary-Action Supermodular Games},
	volume = {92},
	year = {2024}}

@article{chassang2011,
	author = {Chassang, Sylvain and Takahashi, Satoru},
	journal = {Theoretical Economics},
	number = {1},
	pages = {49--93},
	title = {Robustness to incomplete information in repeated games},
	volume = {6},
	year = {2011}}

@article{chen2026,
	author = {Chen, Jaden Yang},
	journal = {American Economic Review},
	number = {1},
	pages = {209--245},
	title = {Sequential Learning under Informational Ambiguity},
	volume = {116},
	year = {2026}}

@unpublished{FII2025,
	author = {Frick, Mira and Iijima, Ryota and Ishii, Yuhta},
	note = {working paper},
	title = {Efficient Learning from Ambiguous Information},
	year = {2025}}

@inproceedings{cheng2024,
	author = {Cheng, Xiaoyu and Klibanoff, Peter and Mukerji, Sujoy and Renou, Ludovic},
	booktitle = {Proceedings of the 25th ACM Conference on Economics and Computation},
	pages = {1201--1202},
	title = {Persuasion with ambiguous communication},
	year = {2024}}

@article{schmeidler1989,
	author = {Schmeidler, David},
	journal = {Econometrica},
	pages = {571--587},
	title = {Subjective probability and expected utility without additivity},
	year = {1989}}

@article{salo1995,
	author = {Salo, Ahtia and Weber, Martin},
	journal = {Journal of Risk and Uncertainty},
	pages = {123--137},
	title = {Ambiguity aversion in first-price sealed-bid auctions},
	volume = {11},
	year = {1995}}

@article{oyama2012,
	author = {Oyama, Daisuke and Tercieux, Olivier},
	journal = {Games and Economic Behavior},
	number = {1},
	pages = {321--331},
	title = {On the strategic impact of an event under non-common priors},
	volume = {74},
	year = {2012}}

@article{dow1992,
	author = {Dow, James and Werlang, Sergio Ribeiro},
	journal = {Econometrica},
	pages = {197--204},
	title = {Uncertainty aversion, risk aversion, and the optimal choice of portfolio},
	year = {1992}}

@article{oyama2010,
	author = {Oyama, Daisuke and Tercieux, Olivier},
	journal = {Journal of Economic Theory},
	number = {2},
	pages = {752--784},
	title = {Robust equilibria under non-common priors},
	volume = {145},
	year = {2010}}

@unpublished{yokota2024,
	author = {Yokota, Yoshifumi},
	note = {working paper},
	title = {Rationalizability in Regular Preference Form Games: Incomplete Information and Higher Order Uncertainty},
	year = {2024}}

@unpublished{cerreia2022,
	author = {Cerreia-Vioglio, Simone and Corrao, Roberto and Lanzani, Giacomo},
	note = {working paper},
	title = {(Un-) Common Preferences, Ambiguity, and Coordination},
	year = {2022}}

@article{huo2024,
	author = {Huo, Zhen and Pedroni, Marcelo and Pei, Guangyu},
	journal = {American Economic Review},
	number = {12},
	pages = {4091--4133},
	title = {Bias and sensitivity under ambiguity},
	volume = {114},
	year = {2024}}

@article{ui2025,
	author = {Ui, Takashi},
	journal = {Games and Economic Behavior},
	pages = {65--81},
	title = {Strategic ambiguity in global games},
	volume = {149},
	year = {2025}}

@article{sakovics2012,
	author = {Sakovics, Jozsef and Steiner, Jakub},
	journal = {American Economic Review},
	number = {7},
	pages = {3439--3461},
	title = {Who matters in coordination problems?},
	volume = {102},
	year = {2012}}

@article{halac2021,
	author = {Halac, Marina and Lipnowski, Elliot and Rappoport, Daniel},
	journal = {American Economic Review},
	number = {3},
	pages = {757--786},
	publisher = {American Economic Association 2014 Broadway, Suite 305, Nashville, TN 37203},
	title = {Rank uncertainty in organizations},
	volume = {111},
	year = {2021}}

@article{segal2003,
	author = {Segal, Ilya},
	journal = {Journal of Economic Theory},
	number = {2},
	pages = {147--181},
	publisher = {Elsevier},
	title = {Coordination and discrimination in contracting with externalities: Divide and conquer?},
	volume = {113},
	year = {2003}}

@unpublished{morris2007,
	author = {Morris, Stephen and Shin, Hyun Song},
	note = {working paper},
	title = {Common belief foundations of global games},
	year = {2007}}

@article{ahn2008,
	author = {Ahn, David S},
	journal = {Journal of Economic Theory},
	number = {1},
	pages = {286--301},
	title = {Hierarchies of ambiguous beliefs},
	volume = {136},
	year = {2007}}

@article{morris2005,
	author = {Morris, Stephen and Ui, Takashi},
	journal = {Journal of Economic Theory},
	number = {1},
	pages = {45--78},
	publisher = {Elsevier},
	title = {Generalized potentials and robust sets of equilibria},
	volume = {124},
	year = {2005}}

@article{ui2001,
	author = {Ui, Takashi},
	journal = {Econometrica},
	number = {5},
	pages = {1373--1380},
	title = {Robust equilibria of potential games},
	volume = {69},
	year = {2001}}

@article{tillio2016,
	author = {di Tillio, Alfredo and Kos, Nenad and Messner, Matthias},
	journal = {Review of Economic Studies},
	number = {1},
	pages = {237--276},
	publisher = {Review of Economic Studies Ltd},
	title = {The design of ambiguous mechanisms},
	volume = {84},
	year = {2016}}

@article{oyama2009,
	author = {Oyama, Daisuke and Tercieux, Olivier},
	journal = {Journal of Economic Theory},
	number = {4},
	pages = {1726--1769},
	publisher = {Elsevier},
	title = {Iterated potential and robustness of equilibria},
	volume = {144},
	year = {2009}}

@article{oyama2020,
	author = {Oyama, Daisuke and Takahashi, Satoru},
	journal = {Econometrica},
	number = {2},
	pages = {693--726},
	publisher = {Wiley Online Library},
	title = {Generalized Belief Operator and Robustness in Binary-Action Supermodular Games},
	volume = {88},
	year = {2020}}

@article{bergemann2019,
	author = {Bergemann, Dirk and Morris, Stephen},
	journal = {Journal of Economic Literature},
	number = {1},
	pages = {44--95},
	title = {Information design: A unified perspective},
	volume = {57},
	year = {2019}}

@article{morris2002,
	author = {Morris, Stephen and Shin, Hyun Song},
	journal = {American Economic Review},
	number = {5},
	pages = {1521--1534},
	title = {Social value of public information},
	volume = {92},
	year = {2002}}

@article{weinstein2007,
	author = {Weinstein, Jonathan and Yildiz, Muhamet},
	journal = {Econometrica},
	number = {2},
	pages = {365--400},
	title = {A structure theorem for rationalizability with application to robust predictions of refinements},
	volume = {75},
	year = {2007}}

@article{rubinstein1989,
	author = {Rubinstein, Ariel},
	journal = {The American Economic Review},
	pages = {385--391},
	publisher = {JSTOR},
	title = {The Electronic Mail Game: Strategic Behavior Under `Almost Common Knowledge'},
	year = {1989}}

@article{dutting2024,
	author = {D{\"u}tting, Paul and Feldman, Michal and Peretz, Daniel and Samuelson, Larry},
	journal = {Econometrica},
	number = {6},
	pages = {1967--1992},
	publisher = {Wiley Online Library},
	title = {Ambiguous contracts},
	volume = {92},
	year = {2024}}

@article{kajii2005,
	author = {Kajii, Atsushi and Ui, Takashi},
	journal = {Japanese Economic Review},
	number = {3},
	pages = {332--351},
	publisher = {Wiley Online Library},
	title = {Incomplete information games with multiple priors},
	volume = {56},
	year = {2005}}

@article{epstein1996,
	author = {Epstein, Larry G and Wang, Tan},
	journal = {Econometrica},
	pages = {1343--1373},
	publisher = {JSTOR},
	volume = {64},
	number = {6},
	title = {`Beliefs about beliefs' without probabilities},
	year = {1996}}

@article{kajii1997,
	author = {Kajii, Atsushi and Morris, Stephen},
	journal = {Econometrica},
	pages = {1283--1309},
	publisher = {JSTOR},
	title = {The robustness of equilibria to incomplete information},
	year = {1997}}

@article{monderer1989,
	author = {Monderer, Dov and Samet, Dov},
	journal = {Games and Economic Behavior},
	number = {2},
	pages = {170--190},
	publisher = {Elsevier},
	title = {Approximating common knowledge with common beliefs},
	volume = {1},
	year = {1989}}

@inproceedings{ganguli2007,
	author = {Ganguli, Jayant V},
	booktitle = {Proceedings of the 11th Conference on Theoretical Aspects of Rationality and Knowledge},
	pages = {155--159},
	title = {Common p-belief and uncertainty},
	year = {2007}}

@article{cheng2019,
	author = {Cheng, Xiaoyu},
	journal = {Journal of Mathematical Economics},
	pages = {102587},
	publisher = {Elsevier},
	title = {Relative maximum likelihood updating of ambiguous beliefs},
	volume = {99},
	year = {2022}}

@article{shishkin2023,
	author = {Shishkin, Denis and Ortoleva, Pietro},
	journal = {Journal of Economic Theory},
	pages = {105610},
	publisher = {Elsevier},
	title = {Ambiguous information and dilation: An experiment},
	volume = {208},
	year = {2023}}

@article{SS2024,
	author = {Sadowski, Philipp and Sarver, Todd},
	journal = {American Economic Review: Insights},
	title = {An Evolutionary Perspective on Updating Risk and Ambiguity Preferences},
	volume = {8},
	number = {2},
	pages = {214--32},
	year = {2026}}

@article{billot2020,
	author = {Billot, Antoine and Mukerji, Sujoy and Tallon, Jean-Marc},
	journal = {Revue {\'e}conomique},
	number = {2},
	pages = {267--282},
	publisher = {Presses de Sciences Po},
	title = {Market allocations under ambiguity: A survey},
	volume = {71},
	year = {2020}}

@article{kellner2018,
	author = {Kellner, Christian and Le Quement, Mark T},
	journal = {Journal of Economic Theory},
	pages = {1--17},
	publisher = {Elsevier},
	title = {Endogenous ambiguity in cheap talk},
	volume = {173},
	year = {2018}}

@article{ellis2016,
	author = {Ellis, Andrew},
	journal = {Theoretical Economics},
	number = {3},
	pages = {865--895},
	publisher = {Wiley Online Library},
	title = {Condorcet meets Ellsberg},
	volume = {11},
	year = {2016}}

@article{klibanoff2005,
	author = {Klibanoff, Peter and Marinacci, Massimo and Mukerji, Sujoy},
	journal = {Econometrica},
	number = {6},
	pages = {1849--1892},
	publisher = {Wiley Online Library},
	title = {A smooth model of decision making under ambiguity},
	volume = {73},
	year = {2005}}

@article{siniscalchi2009,
	author = {Siniscalchi, Marciano},
	journal = {Economics \& Philosophy},
	number = {3},
	pages = {335--356},
	title = {Two out of three ain't bad: A comment on ``The ambiguity aversion literature: A critical assessment''},
	volume = {25},
	year = {2009}}

@article{gul2021,
	author = {Gul, Faruk and Pesendorfer, Wolfgang},
	journal = {Journal of Economic Theory},
	pages = {105129},
	title = {Evaluating ambiguous random variables from Choquet to maxmin expected utility},
	volume = {192},
	year = {2021}}

@article{bose2014,
	author = {Bose, Subir and Renou, Ludovic},
	journal = {Econometrica},
	number = {5},
	pages = {1853--1872},
	title = {Mechanism design with ambiguous communication devices},
	volume = {82},
	year = {2014}}

@article{beauchene2019,
	author = {Beauch{\^e}ne, Dorian and Li, Jian and Li, Ming},
	journal = {Journal of Economic Theory},
	pages = {312--365},
	title = {Ambiguous persuasion},
	volume = {179},
	year = {2019}}

@article{GS1989,
	author = {Gilboa, Itzhak and Schmeidler, David},
	journal = {Journal of Mathematical Economics},
	number = {2},
	pages = {141--153},
	publisher = {Elsevier},
	title = {Maxmin expected utility with non-unique prior},
	volume = {18},
	year = {1989}}

@article{pires2002,
	author = {Pires, Cesaltina Pacheco},
	journal = {Theory and Decision},
	number = {2},
	pages = {137--152},
	title = {A rule for updating ambiguous beliefs},
	volume = {53},
	year = {2002}}

@article{gilboa1993,
	author = {Gilboa, Itzhak and Schmeidler, David},
	journal = {Journal of Economic Theory},
	number = {1},
	pages = {33--49},
	title = {Updating ambiguous beliefs},
	volume = {59},
	year = {1993}}

@article{epstein2007,
	author = {Epstein, Larry G and Schneider, Martin},
	journal = {Review of Economic Studies},
	number = {4},
	pages = {1275--1303},
	title = {Learning under ambiguity},
	volume = {74},
	year = {2007}}

@article{hanany2007,
	author = {Hanany, Eran and Klibanoff, Peter},
	journal = {Theoretical Economics},
	number = {3},
	pages = {261--298},
	title = {Updating preferences with multiple priors},
	volume = {2},
	year = {2007}}
\end{spacing}
\newpage
\setcounter{page}{1}

\singlespacing

\begin{center}
 {\Large\textbf{Online Appendix to ``Contagious Ambiguity''\\[0.7cm]}}
 {\large Mira Frick, Ryota Iijima, and Daisuke Oyama\\[1cm]}
\end{center}

\normalsize

\section{More General Ambiguity Preferences}\label{app:general ambiguity}
\subsection{$\alpha$-Maxmin}\label{app:alpha-MEU}

We generalize the main setting by considering $\alpha$-maxmin expected utility, which allows for a mix of ambiguity-averse and ambiguity-seeking tendencies. Specifically, fix some $\alpha_i \in [0, 1]$ for each $i\in I$. Under any updating rule satisfying Assumption~\ref{asp:ML}, define type $t_i$'s interim payoff to playing action $a_i$ against opponent strategy $\sigma_{-i}$ by
\begin{align*} V_{i} (a_i, \sigma_{-i}, t_i) = &\alpha_i\min_{P_i\in {\cal P}_{t_i}} \sum_{t_{-i}\in T_{-i}} P_i (t_{-i}) v_i(a_i, \sigma_{-i}(t_{-i}), t_i, t_{-i}) \\
+ &(1-\alpha_i)\max_{P_i\in {\cal P}_{t_i}} \sum_{t_{-i}\in T_{-i}} P_i (t_{-i}) v_i(a_i, \sigma_{-i}(t_{-i}), t_i, t_{-i}).
\end{align*}

The following result generalizes Theorem~\ref{prop:contagious}. We omit the proof as the argument is analogous. 

\begin{prop}\label{prop:alpha-maxmin}
Consider players who have $\alpha$-maxmin preferences. Then action profile $a^*\in A$ is contagious if for each $i\in I$ and $a_i\not=a^*_i$,
\begin{equation}\label{eq:alpha-contagion}
\alpha_i\min_{a_{-i}}u_i(a^*_i, a_{-i})+(1-\alpha_i)u_i(a^*)>\alpha_i u_i(a_i, a^*_{-i})+(1-\alpha_i)\max_{a_{-i}}u_i(a_i, a_{-i}).
\end{equation}
\end{prop}

Condition (\ref{eq:alpha-contagion})  reduces to security under maxmin expected utility ($\alpha_i = 1$ for every $i$).
In contrast, for maxmax expected utility ($\alpha_i = 0$ for every $i$), the inequality takes the form $u_i(a^*)>\max_{a_{-i}, a_i \neq a^*_i}u_i(a_i, a_{-i})$. The latter condition is satisfied if, for example, $a^*$ is a strict Nash equilibrium and every action $a_i \neq a^*_i$ is safe. Thus, if agents are ambiguity-seeking, the presence of safe actions serves as a source of contagion for non-safe equilibria. However, as in the maxmin case, this means that almost dominated actions can be contagious. For instance, consider the illustrative example from Section~\ref{sec:example}, where for each $i=1,2$, $A_i=\{0,1\}$ and $u_i(a)={\bf 1}_{\{a_i=a_{-i} =1\}}-ca_i$ for some $c\in (0,1)$.  Action profile $\bf 1$ is contagious under maxmax expected utility. However, $\bf 0$ is $p$-dominant for any $p>1-c$, and hence almost dominant when $c\approx 1$.

\subsection{Second-order Prior}\label{app:smooth}
In this section, we formalize ambiguity using a second-order prior, following the smooth ambiguity model  of \cite{klibanoff2005}. A \textbf{\textit{KMM type space}} $(T,\mu,\phi)$ consists of a countable set $T=\prod_{i\in I} T_i$ of type profiles, a common second-order prior $\mu\in\Delta(\Delta (T))$, and a tuple $\phi=(\phi_i)_{i\in I}$ of aggregators. Here, $\mu$ captures uncertainty about first-order priors $P\in\Delta(T)$ in a Bayesian manner. Each $\phi_i:\mathbb{R}\to\mathbb{R}$ is a monotone transformation that captures player $i$'s attitude toward uncertainty over first-order priors. The analysis below does not rely on the possibility that the $\phi_i$ differ across players.

At type $t_i$,  player $i$ updates the prior $\mu\in\Delta (\Delta (T))$ to $\mu_{t_i}\in\Delta (\Delta (T_{-i}))$. As in \cite{klibanoff2009}, one common updating rule is to define $\mu_{t_i}$ as the pushforward measure of $\tilde\mu_{t_i}\in\Delta (\Delta (T))$ with respect to the mapping $P\mapsto P(\cdot|t_i)$, 
where $\tilde\mu_{t_i}$ is constructed by 
\[
d\tilde\mu_{t_i}(P)=\frac{P(t_i)}{\int_{\Delta(T)}P'(t_i)d\mu(P')}d\mu(P).
\]
Here, Bayes' rule is applied to each first-order prior $P$, as well as the second-order prior $\mu$.\footnote{Our results remain valid under an alternative updating procedure that applies Bayes' rule only to first-order priors, i.e., $\mu_{t_i}(\{P(\cdot \mid t_i): P\in{\cal P}\})=\mu({\cal P})$ for each $t_i$ and measurable ${\cal P}\subseteq\Delta(T)$.} 
We assume that $\int_{\Delta(T)}P'(t_i)d\mu(P')$ is positive for each $t_i$, so that the above construction is well-defined.

Define the $p$-belief operator by 
\[
B^{p}_i(E)=\left\{t_i\in E_i: \int_{\Delta(T_{-i})}\phi_i(P_i(E_{-i})) d\mu_{t_i}(P_i) \geq \phi_i(p)\right\}
\]
for each $i\in I$, $p\in (0,1)$, and product event $E$, and define $B^p(E)$ and $C^p(E)$ analogously to the main setting.\footnote{The belief operator naturally depends on $\phi_i$. As in the unambiguous case, $B_i^p(E)$ corresponds to the set of player $i$'s types in $E_i$ who prefer to receive a bet that pays utility $1$ on $E_{-i}$ and utility $0$ otherwise, rather than a lottery that pays utility $1$ with probability $p$ and utility $0$ with the remaining probability. This is a standard method for eliciting an agent's subjective belief about $E_{-i}$.}

The following result provides an analog of Theorem~\ref{thm:CP-amb} for KMM type spaces. In particular, common $\varepsilon$-belief in an event $E$ can fail even when the ex-ante probability of $E$ is arbitrarily high, and even when the event itself is unambiguous.

\begin{prop}\label{prop:KMM-CP}
For every $\varepsilon \in (0, 1)$, there is a KMM type space $(T, \mu, \phi)$ and a product event $E$ such that $P(E)=P'(E)\geq 1-\varepsilon$ for all $P,P'\in {\rm supp}(\mu)$ but $C^{\varepsilon}(E)=\emptyset$. 
\end{prop}

An incomplete-information game under smooth ambiguity preferences is specified by a KMM type space $(T, \mu, \phi)$, along with a finite set $A_i$ of actions and a  bounded payoff function $v_i: A\times T\to\mathbb R$ for each $i$.
Player $i$'s interim payoff at $t_i$ is given by 
\[
V_i(a_i, \sigma_{-i}, t_i)= \int  \phi_i\left( \sum_{t_{-i}}P_i(t_{-i})v_i (a_i, \sigma_{-i}(t_{-i}), t_i, t_{-i}) \right) d\mu_{t_i}(P_i),
\]
which is linearly extended to mixed actions $\alpha_i \in \Delta(A_i)$ of player $i$. Equilibrium is defined analogously to the main model.

To address the question of contagion, fix a complete-information game $G$.  An $\varepsilon$-elaboration is an incomplete-information game under smooth ambiguity preferences where $P(\prod_{i\in I}T^*_i )\geq 1-\varepsilon$ for all $P  \in {\rm supp}(\mu)$ and $T^*_i$ is the set of $i$'s normal types. The definition of a contagious action profile $a^*$ is analogous to the main model. The following result is an analog of Theorem~\ref{prop:contagious} for smooth ambiguity preferences:

\begin{prop}\label{prop:KMM}
Consider players who have smooth ambiguity preferences. If an action profile $a^*$ in $G$ is secure, then it is contagious.  

\end{prop}

While Proposition~\ref{prop:KMM} allows the aggregators $(\phi_i)$ to depend on the $\varepsilon$-elaboration, this dependence is not needed for the contagion argument. That is, we can fix aggregators $\phi_i$ uniformly across values of $\varepsilon$, provided that the $\phi_i$ are sufficiently concave. For example, under the widely used CARA specification $\phi_i(v)=-\exp[-\alpha_iv]$, the argument goes through as long as each $\alpha_i$ is sufficiently large.

\subsection{Non-additive Prior}\label{app:CEU}  
In this section, we formalize ambiguity using a non-additive prior, following the Choquet expected utility model of \cite{schmeidler1989}.
We consider a \textbf{\textit{Choquet type space}} $(T,\nu)$, where $T$ is a countable set of type profiles and $\nu$ is a \textbf{\textit{capacity}}, i.e., a mapping $\nu:2^T\to[0,1]$ that is monotone with respect to set inclusion and satisfies $\nu(\emptyset)=0$ and $\nu(T)=1$.
The dual capacity is defined by $\overline{\nu}(E):=1-\nu(T\setminus E)$ for each $E\subseteq T$.
An event $E$ is called unambiguous if $\nu(E)=\overline{\nu}(E)$, i.e., if the capacity is additive over $E$ and its complement $T\setminus E$.

Let $\nu_{t_i}$ denote the updated capacity on $T_{-i}$ conditional on $t_i$.  There are two standard updating rules in the literature. Under \textbf{\textit{Dempster-Shafer updating}}, the conditional capacity is given by 
\[
\nu^{\rm DS}_{t_i}(E_{-i})=\frac{\overline\nu(\{t_i\}\times T_{-i})-\overline\nu(\{t_i\}\times (T_{-i}\setminus E_{-i}))}{\overline\nu(\{t_i\}\times T_{-i})}
\]
for each $E_{-i}\subseteq T_{-i}$. 
Under \textbf{\textit{generalized Bayesian updating}}, 
\[
\nu^{\rm GB}_{t_i}(E_{-i})=\frac{\nu(\{t_i\}\times E_{-i})}{\nu(\{t_i\}\times E_{-i})+\overline \nu(\{t_i\}\times (T_{-i}\setminus E_{-i}))}
\]
for each $E_{-i}\subseteq T_{-i}$.\footnote{The latter is also called the Dempster-Fagin-Halpern rule. See \cite{gilboa1993, eichberger2007, horie2013}  for axiomatic foundations of these two rules.}
We assume the denominators above are positive so that the updating rules are well-defined, and each player follows either of these updating rules.

Define $p$-belief operators by
\[
B^{p}_i(E)=\{t_i\in E_i: \nu_{t_i}(E_{-i}) \geq p\}
\]
for each $i\in I$, $p\in (0,1)$, and product event $E$, and define $B^p(E)$ and $C^p(E)$ analogously to the main setting. 

The following result provides an analog of Theorem~\ref{thm:CP-amb} for Choquet type spaces. In particular, common $\varepsilon$-belief in an event $E$ can fail even when the ex-ante probability of $E$ is arbitrarily high, and even when the event itself is unambiguous.

\begin{prop}\label{prop:CEU-CP}
For every $\varepsilon \in (0, 1)$, there is a Choquet type space $(T, \nu)$ and a product event $E$ such that $\nu(E)=\overline \nu(E)\geq 1-\varepsilon$ but $C^{\varepsilon}(E)=\emptyset$. 
\end{prop}

An incomplete-information game under Choquet expected utility is given by a Choquet type space $(T, \nu)$, along with a finite set $A_i$ of actions and a bounded payoff function $v_i: A\times T\to\mathbb R$ for each $i$. Player $i$'s interim payoff at $t_i$ is given by 
\begin{align*}
V_i(a_i, \sigma_{-i}, t_i)
&= \int_0^\infty \nu_{t_i}(\{t_{-i}:v_i (a_i, \sigma_{-i}(t_{-i}), t_i, t_{-i})\geq x \})\, dx \\
&+\int^0_{-\infty} (\nu_{t_i}(\{t_{-i}:v_i (a_i, \sigma_{-i}(t_{-i}), t_i, t_{-i})\geq x \})-1)\,dx,
\end{align*}
which is linearly extended to mixed actions $\alpha_i \in \Delta(A_i)$ of player $i$. Equilibrium is defined analogously to the main model.

To address the question of contagion, fix a complete-information game $G$. An $\varepsilon$-elaboration is an incomplete-information game under Choquet expected utility with $\nu\left(\prod_{i\in I} T_i^*\right)\ge 1-\varepsilon$ and 
$\overline{\nu}\left(\prod_{i\in I} T_i^*\right)\ge 1-\varepsilon$, 
where $T_i^*$ denotes the set of player $i$'s normal types.
The definition of a contagious action profile $a^*$ is analogous to the main model. The following result is an analog of Theorem~\ref{prop:contagious} for Choquet expected utility:

\begin{prop}\label{prop:CEU-contagion}
Consider players with Choquet expected utility preferences. If an action profile $a^*$ in $G$ is secure, then it is contagious.  
\end{prop}

\subsection{Proofs}
\subsubsection{Proofs for Appendix~\ref{app:smooth}}

\begin{proof}[Proof of Proposition~\ref{prop:KMM-CP}]
Take any $\varepsilon \in (0, 1)$. Choose $\phi_i$ for each $i$ such that 
\[
\phi_i(\varepsilon)>\frac{1}{|I|}\phi_i(0)+\frac{|I|-1}{|I|}\phi_i(1).
\]
For each $i\in I$, take $T_i=\{0,1,\ldots, N\}$ and $P^i$  as in the proof of  Theorem~\ref{thm:CP-amb}. Let $\mu(\{P^i\})=\frac{1}{|I|}$ for each $i\in I$.  
Let $E=\{1,2,\ldots, N\}^I$. Observe that $P^i(E)=P^j(E)\geq 1-\varepsilon$ for all $i,j$.

By construction,  $P^i(\{t_i\}\times T_{-i})=P^j(\{t_i\}\times T_{-i})$ for all $i,j$. Thus, for each $t_i>0$,  $\mu_{t_i}$ assigns probability $\frac{1}{|I|}$ on $P^i(\cdot|t_i)$, where $P^i(t_j=t_i-1 \forall j\not=i|t_i)=1$. 
Hence, for each $t_i=n$, 
\[
\int_{\Delta(T_{-i})}\phi_i(P_i(\{t_{-i}\geq {\bf n}\})) d\mu_{t_i}(P_i) \leq\frac{1}{|I|}\phi_i(0)+\frac{|I|-1}{|I|}\phi_i(1).
\]
Given the choice of each $\phi_i$, an inductive argument shows that, for all $k \in \mathbb{N}$, we have $\left(B^{\varepsilon} \right)^k (E)\subseteq \{k+1,k+2,\ldots, N\}^I$. Therefore, $ C^{\varepsilon}(E)=\emptyset$.
\end{proof}

\begin{proof}[Proof of Proposition~\ref{prop:KMM}]
Since $\min_{a_{-i}}u_i(a^*_i, a_{-i})>u_i(a_i, a_{-i}^*)$ for each  $i\in I$ and $a_i\not=a^*_i$, there exist functions $(\phi_i)_{i\in I}$ such that 
\begin{equation}\label{eq:KMM-BR}
\frac{1}{|I|}\phi_i(u_i(a^*))+\frac{|I|-1}{|I|}\phi_i(\min_{a_{-i}}u_i(a^*_i, a_{-i}))>\frac{1}{|I|}\phi_i(u_i(a_i, a_{-i}^*))+\frac{|I|-1}{|I|}\phi_i(\max_{a_{-i}}u_i(a_i, a_{-i}))
\end{equation}
 for all $i\in I$, $a_i\not=a^*_i$.

Fix any $\varepsilon>0$.  We take the KMM type space as in the proof of Proposition~\ref{prop:KMM-CP}. For each $i$, $T^*_i=\{1,2,\ldots, N\}$ is the set of normal types, while $a^*_i$ is dominant for type $t_i=0$. Note that $P^i(\prod_kT^*_k)=P^j(\prod_kT^*_k)\geq 1-\varepsilon$ for every $i,j$.

By construction,  $P^i(t_i)=P^j(t_i)$ for all $i,j$. Thus, for each $t_i>0$,  $\mu_{t_i}$ assigns probability $\frac{1}{|I|}$ on $P^i(\cdot|t_i)$, where $P^i(t_j=t_i-1 \forall j\not= i  \mid t_i)=1$. Combined  with (\ref{eq:KMM-BR}), an inductive argument shows that $a^*$ is played at every type profile in the unique equilibrium.
\end{proof}

\subsubsection{Proofs for Appendix~\ref{app:CEU}}
\begin{proof}[Proof of Proposition~\ref{prop:CEU-CP}]
Fix any $\varepsilon \in (0,1)$. Take the type space $(T, \cP)$ with multiple priors as constructed in the proof of the second part of Proposition~\ref{prop:CP-ML}. Then define a Choquet  type space $(T, \nu)$ by setting $\nu(E)=\min_{P\in \cP}P(E)$ (and hence $\overline\nu(E)=\max_{P\in{\cal P}}P(E)$)  for each $E\subseteq T$.

Consider any $i$ with type $t_i=n>0$. If $i$ follows Dempster-Shafer updating\footnote{Here ``$<$'' denotes component-wise strict inequality. }, 
\begin{eqnarray*}
\nu_{t_i}(\{t_{-i}\not< {\bf n}\})
&=& \frac{\max_{P\in{\cal P}}P(\{t_i=n\})-\max_{P\in{\cal P}}P(\{t_i=n \text{ and } t_{-i}<{\bf n}\})}{\max_{P\in{\cal P}}P(\{t_i=n\})} \\
&=& \frac{P^i(\{t_i=n\})-P^i(\{t_i=n \text{ and } t_{-i}<{\bf n}\})}{P^i(\{t_i=n\})}  =0.
\end{eqnarray*}
If $i$ follows generalized Bayesian updating, 
\begin{eqnarray*}
\nu_{t_i}(\{t_{-i}\not < {\bf n}\})
&=& \frac{\min_{P\in{\cal P}}P(\{t_i=n \text{ and } t_{-i}\not < {\bf n}\})}{\min_{P\in{\cal P}}P(\{t_i=n \text{ and } t_{-i}\not<{\bf n}\})+\max_{P\in{\cal P}}P(\{t_i=n \text{ and } t_{-i}< {\bf n}\})} \\
&=& \frac{P^i(\{t_i=n \text{ and } t_{-i}\not < {\bf n}\})}{P^i(\{t_i=n \text{ and } t_{-i}\not<{\bf n}\})+P^i(\{t_i=n \text{ and } t_{-i}< {\bf n}\})}  =0.
\end{eqnarray*}

Consider the product event $E= \{1,2,\ldots\}^I$. Then $\nu(E)=\overline \nu(E)\geq 1-\varepsilon$. 
Based on the above calculations, one can verify  inductively that for all $k \in \mathbb{N}$, we have $\left(B^{\varepsilon} \right)^k (E)= \{k+1,k+2,\ldots\}^I$. Therefore $ C^{\varepsilon}(E)=\emptyset$. 
\end{proof}

\begin{proof}[Proof of Proposition~\ref{prop:CEU-contagion}] 
Fix any $\varepsilon \in (0, 1)$. Take an $\varepsilon$-elaboration in which the Choquet type space $(T, \nu)$ is as in the proof of Proposition~\ref{prop:CEU-CP}. For each $i$, let $T^*_i=\{1,2,\ldots\}$ and $a^*_i$ be strictly dominant for type $t_i=0$. 
For type $t_i=1$, the interim payoff of $a^*_i$ is at least $\min_{a_{-i}}u_i(a^*_i, a_{-i})$, while the interim payoff of any $a_i\not=a^*_i$ is at most 
\[
 u_i(a_i, a^*_{-i}) + \nu_{t_i}(\{t_{-i}\not< {\bf 1}\}) (\max_{a_{-i}}u_i(a_i, a_{-i})-u_i(a_i, a^*_{-i}))= u_i(a_i, a^*_{-i}),
\]
 where the equality follows from the calculations in the proof of Proposition~\ref{prop:CEU-CP} and the inequality uses the fact that type $t_j=0$ chooses $a^*_j$  for each $j\not=i$. Therefore, security implies that $a^*_i$ is chosen by type $t_i=1$. Repeating the same argument for type $t_i=n$, using
$\nu_{t_i}(\{t_{-i}\not<{\bf n}\})=0$ and induction on $n$, shows that every type
plays $a_i^*$ in the unique equilibrium.
 \end{proof}

\section{Proof of Proposition~\ref{prop:prob}}\label{app:unambig-impl-pf}

\noindent{\bf Lower bound:}
Fix any unambiguous scheme
$ \mathcal C
 =
 \bigl((T_i)_{i\in I},\{P\},(w_i)_{i\in I}\bigr)$
that uniquely implements $\mathbf 1$. We apply the sequential
best-response partition used in Step~1 of Proposition~4 of
\cite{morris2022}. There is a partition
$
 T
 =
\bigcup_{\gamma\in\Pi}T(\gamma)
$
such that, for every $i\in I$ and every type $t_i$ with positive
marginal probability,
\begin{equation}
 \sum_{\gamma\in\Pi}
 \sum_{\substack{t_{-i}:\\
                  (t_i,t_{-i})\in T(\gamma)}}
 P(t_{-i}\mid t_i)
 \left[
 d_i\bigl(S_{-i}(\gamma)\bigr)+w_i(t_i)
 \right]
 >0.
\label{eq:typewise-sequential-obedience}
\end{equation}
Define
\[
 \rho_i(\gamma\mid t_i)
 =
 \sum_{\substack{t_{-i}:\\
                  (t_i,t_{-i})\in T(\gamma)}}
 P(t_{-i}\mid t_i), \;\; 
 \rho(\gamma)
 =
 \sum_{t\in T(\gamma)}P(t).
\]
For every player $i$,
\begin{equation}
 \rho(\gamma)
 =
 \sum_{t_i\in T_i}
P(\{t_i\}\times T_{-i})\rho_i(\gamma\mid t_i),
 \qquad \gamma\in\Pi.
\label{eq:common-rho}
\end{equation}

Using the additive form of the payment,
\eqref{eq:typewise-sequential-obedience} becomes
\[
 w_i(t_i)
 +
 \sum_{\gamma\in\Pi}
 \rho_i(\gamma\mid t_i)
 d_i\bigl(S_{-i}(\gamma)\bigr)
 >0.
\]
By multiplying this inequality by $P(\{t_i\}\times T_{-i})$, summing over $t_i$, and using \eqref{eq:common-rho}, we obtain
\[
 \bar w_i
 +
 \sum_{\gamma\in\Pi}
 \rho(\gamma)d_i\bigl(S_{-i}(\gamma)\bigr)
 >0,
\label{eq:mean-payment-strict-so}
\]
where $ \bar w_i
 =
 \sum_{t_i\in T_i}P(\{t_i\}\times T_{-i})w_i(t_i)$.

Thus $(\bar w,\rho)$ is feasible for the weak inequalities in
\eqref{eq:prob-cost-sequential-obedience}. Since the expected cost of
$\mathcal C$ is $\sum_i\bar w_i$, every uniquely implementing scheme
has cost at least the right-hand side of \eqref{eq:prob-cost-sequential-obedience}. 

\medskip
\noindent{\bf Upper bound:}
Let $(z^\star,\rho^\star)$ solve
\eqref{eq:prob-cost-sequential-obedience}. Fix $\delta>0$ and set $b_i^\delta=z_i^\star+\delta$. Then
\begin{equation}
 \sum_{\gamma\in\Pi}
 \rho^\star(\gamma)
 \left[
 d_i\bigl(S_{-i}(\gamma)\bigr)+b_i^\delta
 \right]
 \ge\delta
 \qquad\text{for every }i\in I.
\label{eq:strict-so-delta}
\end{equation}
Choose a finite number $\widehat b_i>b_i^\delta$ such that
\begin{equation}
 d_i(\varnothing)+\widehat b_i>0.
\label{eq:dominance-payment}
\end{equation}
By supermodularity, action $1$ is then strictly dominant for a type
receiving $\widehat b_i$.

We use the rank construction from Step~2 of Proposition~4 of
\cite{morris2022}. Let
$
 T_i=\{1,2,\ldots\},
$
and, for each $\gamma\in\Pi$, let
$r_i(\gamma)\in\{1,\ldots,n\}$ be the rank of player $i$ in
$\gamma$, where $n := |I|$. For $m\in\mathbb N$, define
\[
 t_i^{m,\gamma}=m+r_i(\gamma).
\]
For $\eta\in(0,1)$, define the common prior by
\[
P_\eta\bigl(t^{m,\gamma}\bigr)
 =
 \eta(1-\eta)^m\rho^\star(\gamma),
 \qquad
 m\in\mathbb N,\quad \gamma\in\Pi,
\label{eq:rank-prior}
\]
with zero probability on all other type profiles. Types of zero
marginal probability may be deleted from $T_i$. Set
\[
 w_i^\eta(\tau)
 =
 \begin{cases}
 \widehat b_i, & \tau\le n-1,\\
 b_i^\delta,  & \tau\ge n.
 \end{cases}
\label{eq:rank-payments}
\]

By \eqref{eq:strict-so-delta}, for all
sufficiently small $\eta>0$,
\begin{equation}
 \sum_{\gamma\in\Pi}
 (1-\eta)^{n-r_i(\gamma)}
 \rho^\star(\gamma)
 \left[
 d_i\bigl(S_{-i}(\gamma)\bigr)+b_i^\delta
 \right]>0
 \qquad\text{for every }i\in I.
\label{eq:weighted-strict-so}
\end{equation}

Every
type $\tau\le n-1$ strictly prefers action $1$ regardless of
opponents' actions, by \eqref{eq:dominance-payment}. Now take
$\tau\ge n$ and suppose that all types below $\tau$ have already been
shown to have action $1$ as their unique rationalizable action.
Conditional on $t_i=\tau$, the opponents who precede $i$ in
$\gamma$ have types below $\tau$. Even if every remaining opponent
chooses action $0$, supermodularity implies that type $\tau$'s
conditional expected payoff gain from action $1$ is at least
\begin{equation}
 \frac{
 \displaystyle
 \sum_{\gamma\in\Pi}
 (1-\eta)^{n-r_i(\gamma)}
 \rho^\star(\gamma)
 \left[
 d_i\bigl(S_{-i}(\gamma)\bigr)+b_i^\delta
 \right]
 }{
 \displaystyle
 \sum_{\gamma\in\Pi}
 (1-\eta)^{n-r_i(\gamma)}
 \rho^\star(\gamma)
 }
 >0.
\label{eq:rank-induction-gain}
\end{equation}
Induction on $\tau$ therefore eliminates action $0$ for every type. Thus action $1$ is uniquely
rationalizable for every type and, in particular, $\mathbf 1$ is the
unique equilibrium. 

It remains to bound the cost. The probability that all players' types are at least $n$ is
$
(1-\eta)^{n-1}.
$
On this event, each player $i$ receives $b_i^\delta$; off this event, total payment is at most
$\sum_i\widehat b_i$. Hence the expected cost is no more than 
\[
(1-\eta)^{n-1} \sum_{i\in I}b_i^\delta
 +
 (1-(1-\eta)^{n-1})\sum_{i\in I}\widehat b_i.
\]
Letting $\eta\downarrow0$ through values satisfying
\eqref{eq:weighted-strict-so} gives $W_{\mathrm{prob}}^{\mathbf 1}
 \le
 \sum_{i\in I}b_i^\delta$. Since $ \sum_{i\in I}b_i^\delta$ equals the right-hand side of \eqref{eq:prob-cost-sequential-obedience} plus $n\delta$, and  $\delta>0$ is arbitrary, this proves that the right-hand side of \eqref{eq:prob-cost-sequential-obedience} upper-bounds 
$W_{\mathrm{prob}}^{\mathbf 1}$.  \qed

\end{document}